\documentclass[12pt]{article}

\usepackage[T1]{fontenc}
\usepackage[latin9]{inputenc}
\usepackage[a4paper]{geometry}
\usepackage{mathtools}
\usepackage{amstext}
\usepackage{amsthm}
\usepackage{amssymb}
\usepackage{graphicx}
\usepackage{setspace}
\usepackage{esint}

\makeatletter
\numberwithin{equation}{section}
\theoremstyle{plain}
\newtheorem*{thm*}{\protect\theoremname}
\theoremstyle{plain}
\newtheorem{thm}{\protect\theoremname}
\theoremstyle{plain}
\newtheorem{lem}[thm]{\protect\lemmaname}

\usepackage{amsfonts}

\newtheorem*{asm}{Assumption}
\theoremstyle{definition}

\providecommand{\lemmaname}{Lemma}
\providecommand{\theoremname}{Theorem}

\makeatother

\providecommand{\lemmaname}{Lemma}
\providecommand{\theoremname}{Theorem}

\usepackage{datetime}
\newdateformat{monthyeardate}{%
	\monthname[\THEMONTH], \THEYEAR}

\begin{document}
	
\begin{titlepage}
\title{
	\vspace{-1.6cm}
	\textbf{Estimation of Average Derivatives of Latent Regressors: With an Application to Inference on Buffer-Stock Saving}
	\thanks{\protect\doublespacing Financial supports from
		Bush Institute-SMU Economic Growth Initiative Research Grant and SMU
		University Research Council Grant are gratefully acknowledged.}
}
\author{
	Hao Dong \thanks{\protect\doublespacing Hao Dong. Department of Economics, Southern Methodist University, 3300 Dyer
		Street, Dallas, TX 75275, US. Email: haod@smu.edu. Phone: (214) 768-3288.}
	\and
	Yuya Sasaki \thanks{\protect\doublespacing Yuya Sasaki. Department of Economics, Vanderbilt University, 2301 Vanderbilt Place, Nashville, TN 37235-1819, US. Email: yuya.sasaki@vanderbilt.edu. Phone: (615) 343-3016.}
}

\date{\monthyeardate\today}

\maketitle

\begin{abstract}
This paper proposes a density-weighted average derivative estimator
based on two noisy measures of a latent regressor. Both measures have
classical errors with possibly asymmetric distributions. We show that
the proposed estimator achieves the root-$n$ rate of convergence,
and derive its asymptotic normal distribution for statistical inference.
Simulation studies demonstrate excellent small-sample performance
supporting the root-$n$ asymptotic normality. Based on the proposed
estimator, we construct a formal test on the sub-unity of the marginal
propensity to consume out of permanent income (MPCP) under a nonparametric
consumption model and a permanent-transitory model of income dynamics
with nonparametric distribution. Applying the test to four recent
waves of U.S. Panel Study of Income Dynamics (PSID), we reject the
null hypothesis of the unit MPCP in favor of a sub-unit MPCP, supporting
the buffer-stock model of saving. 

\vspace{.15in}
\noindent\textbf{Keywords:} Average derivative, latent variables, income dynamics, consumption \\
\noindent\textbf{JEL Codes:} C14, C23, D31

\end{abstract}

\end{titlepage}

\section{Introduction \label{sec:introduction} }

Rational forward-looking agents should have a unit marginal propensity
to consume out of permanent shocks in income (MPCP). Carroll (2009)
demonstrates that the MPCP is strictly less than one in the context
of the buffer-stock model where inpatient consumers have a standard
precautionary saving motive, and further shows through simulations
across a wide range of structural assumptions that the MPCP ranges
from 0.75 to 0.92. Thus, differentiating between the standard model
and Carroll's buffer-stock model can be accomplished by accessing
whether the MPCP is strictly less than one.

When we take this testable implication of the buffer-stock model to
statistical inference based on empirical data, we encounter a fundamental
issue. Namely, we do not observe the permanent income in data. If
we could observe the permanent income shock $X^{\ast}$, as well as
the consumption growth $Y$, then a statistical inference about the
theory of Carroll boils down to inference about the nonparametric
regression function $g(\cdot)=E[Y|X^{\ast}=\cdot]$. Specifically,
rejecting the null hypothesis that a (weighted) average derivative
of $g$ is greater than or equal to one against the alternative that
it is strictly less than one will provide a statistical support for
the theory of Carroll. Under the unobservability of the permanent
income shock $X^{\ast}$ in data, however, the existing econometric
methods of estimation and inference for (weighted) average derivatives
do not apply. To overcome this, we develop a novel method and theory
of estimation and inference for weighted average derivatives when
the latent regressor $X^{\ast}$ is unobserved, but two noisy measures
of $X^{\ast}$ are available in data, as is the case with the standard
permanent-transitory models of earnings and income dynamics.

This paper, in terms of its technical aspects, belongs to the vast literature
on measurement error models and deconvolution. See books by Carroll,
Ruppert, Stefanski and Crainiceanu (2006), Meister (2009) and Horowitz
(2009) and surveys by Chen, Hong and Nekipelov (2011), Schennach (2016)
and Schennach (2021) for reviews. The literature on deconvolution
started out with the deconvolution kernel density methods under known
error distributions (Carroll and Hall, 1988; Stefanski and Carroll,
1990; Fan, 1991a,b; Bissantz, Dümbgen, Holzmann and Munk, 2007; Bissantz
and Holzmann, 2008; van Es and Gugushvili, 2008; Lounici and Nickl,
2011; Schmidt-Hieber, Munk and Dümbgen, 2013), followed by those under
unknown error distributions (Diggle and Hall, 1993; Horowitz and Markatou,
1996; Neumann and Hössjer, 1997; Efromovich, 1997; Li and Vuong, 1998;
Delaigle, Hall and Meister, 2008; Johannes, 2009; Comte and Lacour,
2011; Kato and Sasaki, 2018; Kato, Sasaki and Ura, 2021). Among the
latter set of papers, Horowitz and Markatou (1996) and Delaigle, Hall
and Meister (2008) use repeated measurements with symmetrically and
identically distributed errors, while Li and Vuong (1998) propose
an alternative estimator based on Kotlarski's lemma (Kotlarski, 1967)
that does not require known error distribution -- also see Bohnomme
and Robin (2010) and Comte and Kappus (2015). Also related is Adusumilli,
Kurisu, Otsu and Whang (2020) who studies distribution function instead
of density function.

These deconvolution kernel density methods extend to methods for nonparametric
errors-in-variables regression. Fan and Truong (1993) and Fan and
Masry (1992) study Nadaraya-Watson estimator under known error distribution,
followed by extensions by Delaigle and Meister (2007), Delaigle, Fan
and Carroll (2009) and Delaigle, Hall and Jamshidi (2015). Delaigle,
Hall and Meister (2008), Adusumilli and Otsu (2018) and Kato and Sasaki
(2019) consider cases of unknown error distribution with symmetrically
and identically distributed errors, while Li (2002), Schennach (2004),
Schennach, White and Chalak (2012), Schennach and Hu (2013) and Hu
and Sasaki (2015) consider cases of unknown error distribution with
repeated measurements. Fan (1995) and Dong, Otsu and Taylor (2021)
study average derivatives of nonparametric errors-in-variables regression
under known error distribution and unknown symmetric error distribution,
respectively. The current paper is closely related to the last two
references in that we are also interested in $\sqrt{n}$-consistent
estimation and inference for average derivatives for the purpose of
the aforementioned statistical inference about the hypothesis of buffer-stock
saving. However, unlike Fan (1995) or Dong, Otsu and Taylor (2021),
we allow for the unknown error distribution to be non-symmetric, in
light of the recent empirical reports that components of earnings
and income have skewed distributions (e.g., Bonhomme and Robin, 2010;
Guvenen, Ozkan and Song, 2014; Hu, Moffitt and Sasaki, 2019; Guvenen,
Karahan, Ozkan and Song, 2021).

Despite the extensive econometric and statistical literature on deconvolution
as summarized in the prior two paragraphs, none of the existing papers
to the best of our knowledge can conduct statistical inference for
the MPCP under an unknown and possibly asymmetric error distribution,
even though the skewness has been reported to be very likely by recent
empirical studies on earnings and income dynamics. Motivated by the
aforementioned economic question concerning the MPCP, therefore, this
paper fills this important void in the deconvolution literature by
proposing novel methods of estimation and inference for average derivatives
of latent regressors whose error distribution can be both unknown
and non-symmetric.

Regarding the application to the inference on buffer-stock saving
motivated by Carroll (2009), there are a couple of related existing
papers. In particular, Blundell, Pistaferri and Preston (2008) and
Arellano, Blundell and Bonhomme (2017) investigate MPCP using empirical
data. Blundell, Pistaferri and Preston (2008) use a linear parametric
consumption model and conduct inference on the parameter that represents
a constant MPCP. In contrast, we use a nonparametric consumption model
with a possibly non-constant MPCP. Arellano, Blundell and Bonhomme
(2017) use a nonlinear model of income dynamics, although their identification
and estimation approach \textit{per se} does not lead to statistical
inference on MPCP. In contrast, at the cost of linear model of income
dynamics, our proposed approach allows for statistical inference on
MPCP under a nonparametric consumption model, and thus nonparametrically
enables hypothesis testing about buffer-stock saving. In this way,
the novel method proposed in this paper complements the existing literature
on the empirical analysis of MPCP.

\section{Methodology \label{sec:methodology}}

\subsection{Motivation and overview \label{subsec:motivation}}

Consider the permanent-transitory model of income dynamics 
\begin{align*}
\iota_{jt} & =\pi_{jt}+\tau_{jt}\\
\pi_{jt} & =\pi_{jt-1}+\eta_{jt}
\end{align*}
where $\iota_{jt}$, $\pi_{jt}$, $\tau_{jt}$ and $\eta_{jt}$ denote
observed log income, latent log permanent income, latent log transitory
income and latent permanent income shock, respectively, of individual
$j$ in year $t$. Under this setup, 
\begin{align}
\underbrace{\iota_{jt}-\iota_{jt-2}}_{=:X_{j}} & =\underbrace{\eta_{jt}}_{=:X_{j}^{\ast}}+\underbrace{\eta_{jt-1}+\tau_{jt}-\tau_{jt-2}}_{=:\epsilon_{j}}\qquad\text{and}\label{eq:income:X}\\
\underbrace{\iota_{jt+1}-\iota_{jt-1}}_{=:W_{j}} & =\underbrace{\eta_{jt}}_{=:X_{j}^{\ast}}+\underbrace{\eta_{jt+1}+\tau_{jt+1}-\tau_{jt-1}}_{=:\nu_{j}}\label{eq:income:W}
\end{align}
hold, where $(X_{j},W_{j})$ is observed and $(X_{j}^{\ast},\epsilon_{j},\nu_{j})$
is unobserved by a researcher. Although not necessary, if we assume
that the shocks, $\eta_{jt-1}$, $\eta_{jt}$, $\eta_{jt+1}$, $\tau_{jt-2}$,
$\tau_{jt-1}$, $\tau_{jt}$, and $\tau_{jt+1}$, are mutually independent,
then the latent components $X_{j}^{\ast}$, $\epsilon_{j}$, and $\nu_{j}$
are also mutually independent.\footnote{\protect\doublespacing Even with more periods of data, if we wish errors like $\epsilon_{j}$
and $\nu_{j}$ to keep sharing no common component, due to the structure
of the permanent-transitory model of income dynamics, we can have
no more than two noisy measures of $X_{j}^{*}$. Hence, the repeated
measurements setting here is fundamental.}

Let $Y_{j}$ denote consumption growth $C_{jt}-C_{jt-1}$ from year
$t-1$ to year $t$. Consider the nonparametric regression function
$g(\cdot)=E[Y_{j}|X_{j}^{\ast}=\cdot]$. The derivative $g^{\prime}$
of this function quantifies the MPCP introduced in Section \ref{sec:introduction}.
The theory of Carroll (2009) is that the buffer-stock model that arises
under inpatient consumers with a standard precautionary saving motive
implies $g^{\prime}<1$ rather than $g^{\prime}=1$. This testable
implication leads to the null and alternative hypotheses 
\begin{align}
H_{0}:\theta_{1}\ge0\qquad\text{and}\qquad H_{1}:\theta_{1}<0,\label{eq:hypotheses}
\end{align}
where $\theta_{1}=E[\{g^{\prime}(X^{\ast})-1\}f(X^{\ast})]$ and $f$
is the density of $X^{*}$. Rejection of $H_{0}$ in favor of $H_{1}$
implies that there is at least some location $x^{\ast}$ in the support
of $X^{*}$ such that $g^{\prime}(x^{\ast})<1$, thus providing statistical
support for the theory of Carroll (2009). 

As we formally present in Section \ref{subsec:estimator} ahead, we
propose to estimate $\theta_{1}$ by
\begin{align}
\hat{\theta}_{1} & =-\frac{2}{n^{2}b_{n}^{3}}\sum_{j=1}^{n}\sum_{k=1}^{n}(Y_{j}-W_{j})\int\mathbb{\hat{K}}\left(\frac{x-X_{j}}{b_{n}}\right)\hat{\mathbb{K}}^{\prime}\left(\frac{x-X_{k}}{b_{n}}\right)dx,\label{eq:thetahat}
\end{align}
where $\hat{\mathbb{K}}(u)=\frac{1}{2\pi}\int e^{-\mathrm{i}tu}\frac{K^{\mathrm{ft}}(t)}{\hat{f}_{\epsilon}^{\mathrm{ft}}(t/b_{n})}dt$,
$\mathrm{i}=\sqrt{-1}$, $K^{\mathrm{ft}}$ is the Fourier transform
of a kernel function $K$, $b_{n}$ is a bandwidth parameter, and
$\hat{f}_{\epsilon}^{{\rm ft}}(t)=\frac{\hat{f}_{X}^{{\rm ft}}(t)}{\hat{f}^{{\rm ft}}(t)}$
with $\hat{f}_{X}^{{\rm ft}}(t)=\frac{1}{n}\sum_{j=1}^{n}e^{{\rm i}tX_{j}}$
and $\hat{f}^{{\rm ft}}(t)=\exp\left(\int_{0}^{t}\frac{{\rm i}\sum_{j=1}^{n}X_{j}e^{{\rm i}sW_{j}}}{\sum_{j=1}^{n}e^{{\rm i}sW_{j}}}ds\right)$.
In Section \ref{sec:Main-result}, we further show that $\sqrt{n}(\hat{\theta}_{1}-\theta_{1})$
converges to a normal distribution, which facilitates statistical
inference for $\theta_{1}$. Given the estimator $\hat{\theta}_{1}$
along with its estimated standard error, we may conduct a formal statistical
test of the implication (\ref{eq:hypotheses}) of the theory of Carroll
(2009). 

\subsection{Average derivative estimator\label{subsec:estimator}}

To understand $\hat{\theta}_{1}$ defined in (\ref{eq:thetahat})
in a general framework, we consider the estimation of $\theta_{c}=E[\{g^{\prime}(X^{*})-c\}f(X^{*})]$
for a constant $c$. We set $c=1$ for the test of (\ref{eq:hypotheses}),
whereas one can set $c=0$ if inference for the average derivative
is the objective \emph{per se}. Suppose that $f$ is continuously
differentiable and $\{g(x)-cx\}f^{2}(x)\to0$ as $|x|\to\infty$.
By integration by parts, $\theta_{c}$ can be expressed by
\[
\theta_{c}=-2E[(Y-cX^{*})f^{\prime}(X^{*})].
\]

If $X^{*}$ were directly observed, $\theta_{c}$ could be estimated
by $\tilde{\theta}_{c}=-\frac{2}{n}\sum_{j=1}^{n}(Y_{j}-cX_{j}^{*})\tilde{f}^{\prime}(X_{j}^{*})$
, where $\tilde{f}(x)$ denotes the kernel density estimator of $f(x)$,
and $\tilde{\theta}_{c}$ could be understood as the density-weighted
average derivative estimator by Powell, Stock and Stoker (1989) with
dependent variable $Y-cX^{*}$. In the case when $X^{*}$ is unobserved,
however, $\tilde{\theta}_{c}$ is infeasible. 

Motivated by \eqref{eq:income:X}--\eqref{eq:income:W}, suppose
that we can observe two noisy measurements of $X^{*}$, denoted by
$X$ and $W$, which are generated by
\begin{equation}
X=X^{*}+\epsilon,\quad W=X^{*}+\nu,\label{eq:model}
\end{equation}
where $\epsilon$ and $\nu$ are measurement errors associated with
$X$ and $W$, respectively. In particular, $\epsilon$ and $\nu$
have zero mean and are classical; that is, $\epsilon$ and $\nu$
are independent of $X^{*}$.\footnote{\protect\doublespacing We will invoke a weaker assumption than the full statistical independence
formally in Section \ref{sec:Main-result}.} To construct an estimator of $\theta_{c}$ in this case, note that
\begin{equation}
\theta_{c}=-2E[(Y-cW)f^{\prime}(X^{*})]=-2\int h_{c}(x)f^{\prime}(x)dx,\label{eq:theta}
\end{equation}
where $h_{c}(x)=\{g(x)-cx\}f(x)$. Let $f_{A}$ denote the density
of a random variable $A$, $a^{{\rm ft}}(t)=\int e^{{\rm i}tx}a(x)dx$
denote the Fourier transform of a function $a$, and $\{Y_{j},X_{j},W_{j}\}_{j=1}^{n}$
be an i.i.d. sample of $(Y,X,W)$. 

If $f_{\epsilon}$ were known, $h_{c}$ and $f$ could be estimated
using the deconvolution techniques by
\begin{align*}
\check{h}_{c}(x) & =\frac{1}{nb_{n}}\sum_{j=1}^{n}\mathbb{K}\left(\frac{x-X_{j}}{b_{n}}\right)(Y_{j}-cW_{j}),
\end{align*}
\[
\check{f}(x)=\frac{1}{nb_{n}}\sum_{j=1}^{n}\mathbb{K}\left(\frac{x-X_{j}}{b_{n}}\right),
\]
where $\mathbb{K}(u)=\frac{1}{2\pi}\int e^{-\mathrm{i}tu}\frac{K^{\mathrm{ft}}(t)}{f_{\epsilon}^{\mathrm{ft}}(t/b_{n})}dt$
is the deconvolution kernel function based on the characteristic function
$f_{\epsilon}^{{\rm ft}}$ of the true measurement error $\epsilon$.
Hence, it would be natural to estimate $\theta_{c}$ by the following
plug-in estimator
\begin{align*}
\check{\theta}_{c} & =-2\int\check{h}_{c}(x)\check{f}^{\prime}(x)dx\\
 & =-\frac{2}{n^{2}b_{n}^{3}}\sum_{j=1}^{n}\sum_{k=1}^{n}(Y_{j}-cW_{j})\int\mathbb{K}\left(\frac{x-X_{j}}{b_{n}}\right)\mathbb{K}^{\prime}\left(\frac{x-X_{k}}{b_{n}}\right)dx.
\end{align*}

Now, suppose that $f_{\epsilon}$ is unknown. Assume that $f^{{\rm ft}}$,
$f_{\epsilon}^{{\rm ft}}$ and $f_{\nu}^{{\rm ft}}$ do not vanish
anywhere, according to Kotlarski's (1967) identity, $f^{{\rm ft}}(t)=\exp\left(\int_{0}^{t}\frac{{\rm i}E[Xe^{{\rm i}sW}]}{E[e^{{\rm i}sW}]}ds\right)$,
which together with $f_{\epsilon}^{{\rm ft}}(t)=\frac{E[e^{{\rm i}tX}]}{f^{{\rm ft}}(t)}$
implies that $\hat{f}_{\epsilon}^{{\rm ft}}$ defined in Section \ref{subsec:motivation}
is a plug-in estimator of $f_{\epsilon}^{{\rm ft}}$ based on the
sample analogs of $E[e^{{\rm i}sX}]$ and $E[e^{{\rm i}tW}]$, and
$E[Xe^{{\rm i}sW}]$. Therefore, to estimate $\theta_{c}$ when $f_{\epsilon}$
is unknown, it is natural to replace $f_{\epsilon}^{{\rm ft}}$ in
$\check{\theta}_{c}$ by $\hat{f}_{\epsilon}^{{\rm ft}}$, which gives
\[
\hat{\theta}_{c}=-\frac{2}{n^{2}b_{n}^{3}}\sum_{j=1}^{n}\sum_{k=1}^{n}(Y_{j}-cW_{j})\int\mathbb{\hat{K}}\left(\frac{x-X_{j}}{b_{n}}\right)\hat{\mathbb{K}}^{\prime}\left(\frac{x-X_{k}}{b_{n}}\right)dx,
\]
where $\hat{\mathbb{K}}(u)=\frac{1}{2\pi}\int e^{-\mathrm{i}tu}\frac{K^{\mathrm{ft}}(t)}{\hat{f}_{\epsilon}^{\mathrm{ft}}(t/b_{n})}dt$,
defined in Section \ref{subsec:motivation}, is the deconvolution
kernel function based on the estimated measurement error characteristic
function $\hat{f}_{\epsilon}^{{\rm ft}}$. 

\section{Main result \label{sec:Main-result}}

In this section, we presents a formal theory that provides the asymptotic
validity of the test procedure proposed in Section \ref{sec:methodology}.
Specifically, we derive the asymptotic distribution for $\hat{\theta}_{c}$
and propose an estimator for its asymptotic variance. To this end,
we make the following assumptions.


\begin{asm}\quad 
\begin{description}
\item [{(1)}] $\{Y_{j},X_{j},W_{j}\}_{j=1}^{n}$ is an i.i.d. sample of
$(Y,X,W)$, where $(X,W)$ satisfies (\ref{eq:model}), $E[|X^{*}|^{2+\eta}]<\infty$
for some $\eta>0$, and $E[Y^{2}]<\infty$. Measurement errors $(\epsilon,\nu)$
are independent from $X^{*}$ and satisfy $E[Y|X^{*},\epsilon]=E[Y|X^{*}]$,
$E[\epsilon|\nu]=0$, $E[\nu|\epsilon]=0$, $E[|\epsilon|^{2+\eta}]<\infty$,
and $E[|\nu|^{2+\eta}]<\infty$. Characteristic functions $f^{{\rm ft}}$,
$f_{\epsilon}^{{\rm ft}}$ and $f_{\nu}^{{\rm ft}}$ do not vanish
anywhere.
\item [{(2)}] $h_{c}$ and $f$ have $\alpha$ continuous, bounded and
integrable derivatives, and satisfy
\begin{align*}
|f^{(\alpha)}(x+\Delta x)-f^{(\alpha)}(x)| & <m(x)|\Delta x|,\quad|h_{c}^{(\alpha)}(x+\Delta x)-h_{c}^{(\alpha)}(x)|<m(x)|\Delta x|
\end{align*}
for some bounded and integrable function $m(x)$ with $E[|m(X)|^{2}(1+|Y-cW|)^{2}]<\infty$.
\item [{(3)}] $K$ is symmetric, differentiable, and$\int K(u)du=1$, $\int u^{l}K(u)du=0$
for $1\le l<\alpha$, and $\int u^{\alpha}K(u)du\ne0$. Also $K^{{\rm ft}}$
is compactly supported on $[-1,1]$ and bounded.
\item [{(4)}] $\frac{n^{-1/2}b_{n}^{-2}\log(1/b_{n})^{2}\{\inf_{|t|\le b_{n}^{-1}}|f_{\epsilon}^{{\rm ft}}(t)|\}^{-2}\{\inf_{|t|\le b_{n}^{-1}}|f^{{\rm ft}}(t)|\}^{-2}}{\min\left\{ \{\inf_{|t|\le b_{n}^{-1}}|f_{\epsilon}^{{\rm ft}}(t)|\}^{2},\{\inf_{|t|\le b_{n}^{-1}}|f_{\nu}^{{\rm ft}}(t)|\}^{4}\{\inf_{|t|\le b_{n}^{-1}}|f^{{\rm ft}}(t)|\}^{2}b_{n}^{2}\right\} }\to0$
and $n^{1/2}b_{n}^{\alpha}\to0$ as $n\to\infty$. 
\item [{(5)}] $Var[\xi_{c}(Y,X,W)]<\infty$ where
\begin{gather*}
\xi_{c}(y,x,w)=\frac{1}{\pi}\int\left\{ \begin{array}{c}
\{\{h_{c}^{\prime}\}^{{\rm ft}}(-t)-(y-cw)\{f^{\prime}\}^{{\rm ft}}(-t)\}\frac{e^{{\rm i}tx}}{f_{\epsilon}^{{\rm ft}}(t)}\\
+\left\{ \begin{array}{c}
\{f^{{\rm ft}}(t)\{h_{c}^{\prime}\}^{{\rm ft}}(-t)-\{f^{\prime}\}^{{\rm ft}}(-t)h_{c}^{{\rm ft}}(t)\}\\
\times\left\{ -\frac{e^{{\rm i}tx}}{f^{{\rm ft}}(t)f_{\epsilon}^{{\rm ft}}(t)}+\int_{0}^{t}\left(-\frac{\{f^{{\rm ft}}\}^{\prime}(s)}{f^{{\rm ft}}(s)}+{\rm i}x\right)\frac{e^{{\rm i}sw}}{f^{{\rm ft}}(s)f_{\nu}^{{\rm ft}}(s)}ds\right\} 
\end{array}\right\} 
\end{array}\right\} dt.
\end{gather*}
\end{description}
\end{asm}

Assumption (1) requires random sampling and imposes conditions on
the distribution of $(Y,X^{*})$ and measurement errors $(\epsilon,\nu)$.
In particular, under the classical measurement error assumptions,
$E[Y|X^{*},\epsilon]=E[Y|X^{*}]$ and $E[\nu|\epsilon]=E[\nu]$ are
imposed for the identification of $E[Y-cW|X^{*}]$\footnote{\protect\doublespacing The full independence between the measurement error in regressor and
dependent variable has been commonly adopted in the literature of
nonparametric regression with errors-in-variables, (e.g., Fan and
Truong, 1993; Delaigle and Meister, 2007; Meister, 2009). In our setting,
however, it is sufficient to require the mean independence. In particular,
observe that the identification of $E[Y-cW|X^{*}]$ hinges on $$\{E[Y-cW|X]f_{X}\}(x)=\int\{E[Y-cW|X^{*}]f\}(x-e)f_{\epsilon}(e)de,$$ 
which holds if $E[Y-cW|X^{*},\epsilon]=E[Y-cW|X^{*}]$, or equivalently
$E[Y|X^{*},\epsilon]=E[Y|X^{*}]$ and $E[\nu|\epsilon]=E[\nu]$, under
the classical measurement error assumption. }, $E[\nu]=0$ is imposed for $\theta_{c}=-2E[(Y-cW)f^{\prime}(X^{*})]$,
and the non-vanishing characteristic functions and $E[\epsilon|\nu]=0$
is imposed for the Kotlarski's identity\footnote{\protect\doublespacing The full independence between $\epsilon$ and $\nu$ has been commonly
adopted in the literature on identification and estimation based on
Kotlarski's identity (e.g., Li and Vuong, 1998; Kato, Sasaki and Ura,
2021; Dong, Otsu and Taylor, 2022), but it
can be relaxed to the mean independence; See Schennach(2004).}. $E[|X^{*}|^{2+\eta}]<\infty$ and $E[|\epsilon|^{2+\eta}]<\infty$
are regularity conditions required by Lemma \ref{lem:delta}, which
is used to characterize the uniform convergence rate of the empirical
characteristic function of $(X,W)$ and their first-order derivatives
over an expanding region. We remark in the context of (\ref{eq:income:X})--(\ref{eq:income:W})
that our conditions, $E[\epsilon|\nu]=0$ and $E[\nu|\epsilon]=0$,
do not rule out the typical assumptions about the permanent-transitory
models of income and earnings dynamics in which the permanent shocks
$\left\{ \eta_{jt}\right\} _{t}$ and the transitory components $\left\{ \tau_{jt}\right\} _{t}$
are white noise processes (e.g., Bonhomme and Robin, 2010).

Assumption (2) constitutes mild assumptions on the smoothness of the
density function $f$ and the regression function $g$, which is equivalent
to Assumptions 3 and 5 in Powell, Stock and Stoker (1989). Assumption
(3) concerns the kernel function $K$. Specifically, we require $K$
to be a symmetric $\alpha$-th order kernel, which together with Assumption
(2) can be used to control the magnitude of the estimation bias. In
addition, we also require $K^{{\rm ft}}$ to be compactly supported,
which is to regularize the deconvolution problem that is well-known
to be ill-posed. 

Assumption (4) gives two conditions on the bandwidth $b_{n}$. In
particular, the first condition is needed to control the estimation
variance, and the second condition requires $g$ and $f$ to be smooth
enough so that the estimation bias is asymptotically negligible. Here,
we maintain a general expression without specifying the decay rates
of the tails of $f^{{\rm ft}}$, $f_{\epsilon}^{{\rm ft}}$ and $f_{\nu}^{{\rm ft}}$
as is typical in the deconvolution literature. By doing so, we can
apply our result to a larger set of measurement error distributions,
including both ordinary smooth distributions and supersmooth distributions.
Assumption (5) is a high-level assumption on the boundedness of the
asymptotic variance of $\hat{\theta}_{c}$; an analogous assumption
is made in Fan (1995) and Dong, Otsu and Taylor (2021).
\begin{thm*}
Under Assumptions (1) - (5), we have
\[
\sqrt{n}\{\hat{\theta}_{c}-\theta_{c}\}\overset{d}{\to}N(0,Var[\xi_{c}(Y,X,W)]).
\]
\end{thm*}

This is the main result of the paper. To understand it, we can decompose
$\xi_{c}$ into two parts as $\xi_{c}(Y,X,W)=\xi_{c,1}(Y,X,W)+\xi_{c,2}(Y,X,W)$,
where
\begin{gather*}
\xi_{c,1}(Y,X,W)=\frac{1}{\pi}\int\left\{ \{h_{c}^{\prime}\}^{{\rm ft}}(-t)-(Y-cW)\{f^{\prime}\}^{{\rm ft}}(-t)\right\} \frac{e^{{\rm i}tX}}{f_{\epsilon}^{{\rm ft}}(t)}dt,
\\
\xi_{c,2}(Y,X,W)=\frac{1}{\pi}\int\left\{ \begin{array}{c}
	\{\{h_{c}^{\prime}\}^{{\rm ft}}(-t)f^{{\rm ft}}(t)-h_{c}^{{\rm ft}}(t)\{f^{\prime}\}^{{\rm ft}}(-t)\}\\
	\times\left\{ -\frac{e^{{\rm i}tX}}{f^{{\rm ft}}(t)f_{\epsilon}^{{\rm ft}}(t)}+\int_{0}^{t}\left(-\frac{\{f^{{\rm ft}}\}^{\prime}(s)}{f^{{\rm ft}}(s)}+{\rm i}X\right)\frac{e^{{\rm i}sW}}{f^{{\rm ft}}(s)f_{\nu}^{{\rm ft}}(s)}ds\right\} 
\end{array}\right\} dt,
\end{gather*}
and consider, for ease of illustration, a special case in which $f_{\epsilon}^{{\rm ft}}$
is of the form
\[
f_{\epsilon}^{{\rm ft}}(t)=\frac{1}{c_{0}+c_{1}t+\cdots+c_{\beta}t^{\beta}},
\]
where $\beta$ is a positive integer, and $c_{0}=1$, $c_{1},\ldots,c_{\beta}$
are complex numbers, which includes the Laplace distribution as a
special case when $c_{1}=0$ and $\beta=2$. 

In such cases, using $\{a^{(k)}\}^{{\rm ft}}(t)=(-{\rm i}t)^{k}a^{{\rm ft}}(t)$
for a positive integer $k$, we obtain that
\begin{align*}
\xi_{c,1}(Y,X,W)= & \frac{1}{\pi}\int\left\{ \{h_{c}^{\prime}\}^{{\rm ft}}(-t)-(Y-cW)\{f^{\prime}\}^{{\rm ft}}(-t)\right\} \left\{ c_{0}+c_{1}t+\cdots+c_{\beta}t^{\beta}\right\} e^{{\rm i}tX}dt\\
= & \sum_{k=0}^{\beta}\frac{c_{k}}{\pi{\rm i}^{k}}\int\left\{ \underbrace{({\rm i}t)^{k+1}h_{c}^{{\rm ft}}(-t)}_{\{h_{c}^{(k+1)}\}^{{\rm ft}}(-t)}-(Y-cW)\underbrace{({\rm i}t)^{k}\{f^{\prime}\}^{{\rm ft}}(-t)}_{\{f^{(k+1)}\}^{{\rm ft}}(-t)}\right\} e^{{\rm i}tX}dt\\
= & \sum_{k=0}^{\beta}(-{\rm i})^{k}2c_{k}\left\{ (Y-cW)f^{(k+1)}(X)-h_{c}^{(k+1)}(X)\right\} ,
\end{align*}
which when $c=0$ coincides with $2r(X,Y)$ defined as in equation
(25) of Fan (1995). Thus, $\xi_{c,1}(Y,X,W)$ characterizes the randomness
of $\check{\theta}_{c}$, which is the estimator of $\theta_{c}$
when $f_{\epsilon}$ is known. Furthermore, compared to $2r(X,Y)$
in Fan (1995), $\xi_{c,1}(Y,X,W)$ allows non-zero value of $c$ and
can cover a larger set of measurement error distributions. 

Since $\xi_{c,1}(Y,X,W)$ characterizes the randomness in the estimation
of $\theta_{c}$ when $f_{\epsilon}$ is known, the additional randomness
introduced by using $\hat{f}_{\epsilon}^{{\rm ft}}$ in the place
of $f_{\epsilon}^{{\rm ft}}$ is reflected by $\xi_{c,2}(Y,X,W)$.
It is worthy to note that the structure of $\xi_{c,2}(Y,X,W)$ is
similar to that of $\xi_{c,1}(Y,X,W)$, but is more complicated. In
general, it is difficult to simplify $\xi_{c,2}(Y,X,W)$ as we did
for $\xi_{c,1}(Y,X,W)$ even when $f^{{\rm ft}}$, $f_{\epsilon}^{{\rm ft}}$
and $f_{\nu}^{{\rm ft}}$ are all specified. However, there are special
cases in which $\xi_{c,2}(Y,X,W)$ can be completely ignored. To see
a case in point, note that $\{h_{c}^{\prime}\}^{{\rm ft}}(-t)f^{{\rm ft}}(t)-h_{c}^{{\rm ft}}(t)\{f^{\prime}\}^{{\rm ft}}(-t)={\rm i}t\{h_{c}^{{\rm ft}}(-t)f^{{\rm ft}}(t)-h_{c}^{{\rm ft}}(t)f^{{\rm ft}}(-t)\}$,
which implies that if $h_{c}^{{\rm ft}}(-t)=h_{c}^{{\rm ft}}(t)$
and $f^{{\rm ft}}(-t)=f^{{\rm ft}}(t)$, for example when both $h_{c}$
and $f$ are symmetric around zero, $\xi_{c,2}(Y,X,W)=0$, i.e. the
estimation error brought by using $\hat{f}_{\epsilon}^{{\rm ft}}$
in the place of $f_{\epsilon}^{{\rm ft}}$ is exactly zero and $\hat{\theta}_{c}$
has exactly the same asymptotic distribution as that of $\check{\theta}_{c}$.

To test hypotheses on $\theta_{c}$ like $H_{0}$, besides the asymptotic
distribution of $\hat{\theta}_{c}$, we also need to estimate $s_{c}^{2}=Var[\xi_{c}(Y,X,W)]$.
Observe that we can rewrite
\begin{align*}
\xi_{c}(y,x,w) & =\frac{1}{\pi}\int{\rm i}t\left\{ \begin{array}{c}
\{h_{c}^{{\rm ft}}(-t)-(y-cw)f^{{\rm ft}}(-t)\}\frac{e^{{\rm i}tx}}{f_{\epsilon}^{{\rm ft}}(t)}\\
+\left\{ \begin{array}{c}
\{f^{{\rm ft}}(t)h_{c}^{{\rm ft}}(-t)-f^{{\rm ft}}(-t)h_{c}^{{\rm ft}}(t)\}\\
\times\left\{ -\frac{e^{{\rm i}tx}}{f^{{\rm ft}}(t)f_{\epsilon}^{{\rm ft}}(t)}+\int_{0}^{t}\left(-\frac{\{f^{{\rm ft}}\}^{\prime}(s)}{f^{{\rm ft}}(s)}+{\rm i}x\right)\frac{e^{{\rm i}sw}}{f^{{\rm ft}}(s)f_{\nu}^{{\rm ft}}(s)}ds\right\} 
\end{array}\right\} 
\end{array}\right\} dt\\
 & =\frac{1}{\pi}\int{\rm i}t\left\{ \begin{array}{c}
\{h_{c}^{{\rm ft}}(-t)-(y-cw)f^{{\rm ft}}(-t)\}\frac{e^{{\rm i}tx}}{f_{\epsilon}^{{\rm ft}}(t)}\\
+\left\{ \begin{array}{c}
\{f^{{\rm ft}}(t)h_{c}^{{\rm ft}}(-t)-f^{{\rm ft}}(-t)h_{c}^{{\rm ft}}(t)\}\\
\times\left\{ -\frac{e^{{\rm i}tx}}{f_{X}^{{\rm ft}}(t)}+\int_{0}^{t}\left(-\frac{{\rm i}E[Xe^{{\rm i}sW}]}{f_{W}^{{\rm ft}}(s)}+{\rm i}x\right)\frac{e^{{\rm i}sw}}{f_{W}^{{\rm ft}}(s)}ds\right\} 
\end{array}\right\} 
\end{array}\right\} dt,
\end{align*}
where the first step uses $\{a^{(k)}\}^{{\rm ft}}(t)=(-{\rm i}t)^{k}a^{{\rm ft}}(t)$
and the second step follows from $f_{X}^{{\rm ft}}=f^{{\rm ft}}f_{\epsilon}^{{\rm ft}}$,
$f_{W}^{{\rm ft}}=f^{{\rm ft}}f_{\nu}^{{\rm ft}}$ and $f^{{\rm ft}}(t)=\exp\left(\int_{0}^{t}\frac{{\rm i}E[Xe^{{\rm i}sW}]}{E[e^{{\rm i}sW}]}ds\right)$,
which implies $E[\xi_{c}(Y,X,W)]=0$ and $s_{c}^{2}=E[\xi_{c}^{2}(Y,X,W)]$.
In practice,$f^{{\rm ft}}$, $f_{\epsilon}^{{\rm ft}}$, $f_{X}^{{\rm ft}}$,
$f_{W}^{{\rm ft}}$, $E[Xe^{{\rm i}tW}]$ and $h_{c}^{{\rm ft}}$
are all unknown, and we have to estimate them. $f^{{\rm ft}}$, $f_{\epsilon}^{{\rm ft}}$
and $f_{X}^{{\rm ft}}$ can be estimated by $\hat{f}^{{\rm ft}}$,
$\hat{f}_{\epsilon}^{{\rm ft}}$, and $f_{X}^{{\rm ft}}$ defined
in Section \ref{sec:methodology}, $f_{W}^{{\rm ft}}$ can be estimated
by $\hat{f}_{W}^{{\rm ft}}(t)=\frac{1}{n}\sum_{j=1}^{n}e^{{\rm i}tW_{j}}$,
$E[Xe^{{\rm i}tW}]$ can be estimated by $\hat{E}[Xe^{{\rm i}tW}]=\frac{1}{n}\sum_{j=1}^{n}X_{j}e^{{\rm i}sW_{j}}$,
and $h_{c}^{{\rm ft}}$ can be estimated by $\hat{h}_{c}^{{\rm ft}}(t)=\frac{\frac{1}{n}\sum_{j=1}^{n}(Y_{j}-cW_{j})e^{{\rm i}tX_{j}}}{\hat{f}_{\epsilon}^{\mathrm{ft}}(t)}$.
Therefore, we can estimate $s_{c}^{2}$ by $\hat{s}_{c}^{2}=\frac{1}{n}\sum_{j=1}^{n}\hat{\xi}_{c}^{2}(Y_{j},X_{j},W_{j})$
with
\[
\hat{\xi}_{c}(y,x,w)=\frac{{\rm 1}}{\pi}\int{\rm i}t\left\{ \begin{array}{c}
\{\hat{h}_{c}^{{\rm ft}}(-t)-(y-cw)\hat{f}^{{\rm ft}}(-t)\}\frac{e^{{\rm i}tx}}{\hat{f}_{\epsilon}^{{\rm ft}}(t)}\\
+\left\{ \begin{array}{c}
\{\hat{f}^{{\rm ft}}(t)\hat{h}_{c}^{{\rm ft}}(-t)-\hat{f}^{{\rm ft}}(-t)\hat{h}_{c}^{{\rm ft}}(t)\}\\
\times\left\{ -\frac{e^{{\rm i}tx}}{\hat{f}_{X}^{{\rm ft}}(t)}+\int_{0}^{t}\left(-\frac{{\rm i}\hat{E}[Xe^{{\rm i}sW}]}{\hat{f}_{W}^{{\rm ft}}(s)}+{\rm i}x\right)\frac{e^{{\rm i}sw}}{\hat{f}_{W}^{{\rm ft}}(s)}ds\right\} 
\end{array}\right\} 
\end{array}\right\} K^{{\rm ft}}(tb_{n})dt,
\]
where $K^{{\rm ft}}(tb_{n})$ is introduced to regularize the integration. 

\section{Simulation\label{sec:Simulation}}

This section presents simulation studies to analyze the finite sample
performance of the proposed method of inference about $\theta_{c}$.
We generate $N$ independent copies of the observed variables $(Y,X,U)$
via the structural equations
\begin{align*}
Y & =f(X^{\ast})+U=X^{\ast}-\delta X^{\ast}+U,\\
X & =X^{\ast}+\epsilon,\qquad\textnormal{and}\\
W & =X^{\ast}+\nu,
\end{align*}
where the latent variables $(X^{\ast},U,\epsilon,\nu)$ are in turn
generated independently from the standard normal distribution. Note
in this setting that the null hypothesis $H_{0}:\theta_{1}\ge0$ is
true if and only if $\delta\le0$. Furthermore, positive values of
the design parameter $\delta$ in this data generating process measures
deviations from the null hypothesis. We run sets of simulations across
combinations of the values of $\delta\in[0.0,0.5]$ and $N\in\left\{ 250,500\right\} $,
where each set of simulations consists of 2,500 Monte Carlo iterations. 

\begin{center}
	[FIGURE 1 HERE]
\end{center}

Figure \ref{fig:simulation} plots the Monte Carlo frequencies of
rejecting the null hypothesis $H_{0}:\theta_{1}\ge0$ against the
alternative $H_{1}:\theta_{1}<0$ based on the one-sided test with
our estimator $\hat{\theta}_{1}$ and its standard error estimator
$\hat{s}_{1}$. The nominal size of the test is set to 0.05. The horizontal
axis of the figure measures the deviation $\delta\in[0.0,0.5]$ away
from the null hypothesis $H_{0}$. The dashed (respectively, solid)
curve indicates the results with $N=250$ (respectively, $500$).
Observe that the rejection frequency at $\delta=0$ is close to the
nominal size, 0.05. As $\delta$ becomes larger, on the other hand,
the rejection frequencies increase. For a given value of $\delta>0$,
the rejection frequency is larger for the larger sample size, demonstrating
the power of the test as well as the size control. 

We ran additional simulations with alternative data generating designs,
only to find very similar simulation results to the baseline results
presented above. Overall, the simulation outcomes demonstrate excellent
small-sample performance of the estimation and inference methods.
Our observations that the asymptotic approximations are already accurate
even at such small sample sizes $N$ as $250$ demonstrate the practical
merit of our root-n consistent test under the highly sophisticated
problem of errors-in-variables nonparametric regressions.

\section{Application \label{sec:Application}}

This section revisits the analysis of the MPCP introduced in Sections
\ref{sec:introduction}--\ref{sec:methodology}. Recall from (\ref{eq:income:X})
and (\ref{eq:income:W}) that we require four time periods of panel
data $\left\{ \left\{ \iota_{j\tau}\right\} _{\tau=t-2}^{t+1}\right\} _{j=1}^{n}$
on observed income to construct the variables $\left\{ (X_{j},W_{j})\right\} _{j=1}^{n}$
that we use as inputs for our proposed method of inference. Using
the U.S. Panel Study of Income Dynamics (PSID) for the four most recent
survey years 2013, 2015, 2017, and 2019, we aim to test the null hypothesis
$H_{0}:\theta_{1}\ge0$ of (super-) unit MPCP against the alternative
hypothesis $H_{1}:\theta_{1}<0$ of sub-unit MPCP. Rejecting the null
hypothesis $H_{0}$ supports the buffer-stock model that arises with
inpatient consumers having a standard precautionary saving motive. 

The income variable $\iota_{jt}$ is defined by the log total family
income of unit $j$ reported in year $t$. The consumption variable
$C_{jt}$ is similarly defined by the log family expenditures of unit
$j$ reported in year $t$, where the categories of consumption consist
of food, housing, telephone/internet, transportation, vehicle, education,
child care, health care, household repairs, household furnishing,
clothing, and recreation. The first three columns of Table \ref{tab:summary_statistics}
display summary statistics of this data set. Displayed values are
the sample means. Parentheses enclose sample standard deviations.

\begin{center}
	[TABLE 1 HERE]
\end{center}

Adapting (\ref{eq:income:X}) and (\ref{eq:income:W}) to this panel
data set, we construct $X_{j}=\iota_{j2017}-\iota_{j2013}$ and $W_{j}=\iota_{j2019}-\iota_{j2015}$.
Likewise, we construct the outcome variable by $Y_{j}=C_{j2017}-C_{j2015}$.
We drop units that experience a missing value or an infinite value
for $X_{j}$, $W_{j}$ or $Y_{j}$. Consequently, we obtain a balanced
panel of 5976 household units. The last three columns of Table \ref{tab:summary_statistics}
display summary statistics of the constructed variables $\left\{ (X_{j},W_{j},Y_{j})\right\} _{j=1}^{n}$.
Again, displayed values are the sample means, and parentheses enclose
sample standard deviations.

Applying our proposed method of estimation and inference, we obtain
the point estimate of $\hat{\theta}_{1}=-0.0607$ with the estimated
standard error of $\hat{s}_{1}=0.0052$. Using our asymptotic normality
results along with these estimates, we formally reject the one-sided
test of the null hypothesis $H_{0}:\theta_{1}\ge0$ of (super-) unit
MPCP in favor of the alternative hypothesis $H_{1}:\theta_{1}<0$
of sub-unit MPCP. Our test result supports the buffer-stock model
that arises with inpatient consumers having a standard precautionary
saving motive. Even though a number of prior studies have calibrated
or estimated the MPCP under various models of income dynamics, to
our best knowledge, our result is the first formal statistical inference
result about the MPCP using flexible nonparametric distribution in
the permanent-transitory model of income processes.

\section{Conclusion\label{sec:Conclusion}}

In this paper, we propose a density-weighted average derivative estimator
when two noisy measures of a latent regressor is available. Both measures
have classical errors, and the error distributions are possibly asymmetric.
We show that this estimator achieves the root-$n$ rate of convergence
and is asymptotically normal. Simulation studies demonstrate excellent
small-sample performance, and support the merit of the root-$n$ asymptotic
normality. Based on the proposed estimator, we construct a test on
the sub-unity of MPCP under a nonparametric consumption model. In
particular, under a permanent-transitory model of income dynamics,
we construct two noisy measures of a permanent income shock using
four periods data. With an application using recent waves of the U.S.
PSID, we reject the null hypothesis of unit MPCP in favor or a sub-unit
MPCP, supporting the buffer-stock model of saving.

\newpage{}

\appendix

\section{Proof of the theorem}

Let $\hat{\mu}_{\iota}(t)=\frac{1}{n}\sum_{l=1}^{n}\mu_{\iota,l}(t)$
and $\mu_{\iota}(t)=E[\mu_{\iota,1}(t)]$ for $\iota=1,2,3$ with
$\mu_{1,l}(t)=e^{{\rm i}tX_{l}}$, $\mu_{2,l}(t)=e^{{\rm i}tW_{l}}$,
and $\mu_{3,l}(t)=X_{l}e^{{\rm i}tW_{l}}$. Then, $\hat{f}_{\epsilon}^{{\rm ft}}(t)=\hat{\mu}_{1}(t)\exp\left(-\int_{0}^{t}\frac{{\rm i}\hat{\mu}_{3}(s)}{\hat{\mu}_{2}(s)}ds\right)$
and $f_{\epsilon}(t)=\mu_{1}(t)\exp\left(-\int_{0}^{t}\frac{{\rm i}\mu_{3}(s)}{\mu_{2}(s)}ds\right)$.
By expanding $(\hat{\mu}_{1},\hat{\mu}_{2},\hat{\mu}_{3})$ around
$(\mu_{1},\mu_{2},\mu_{3})$, we obtain
\begin{equation}
\mathbb{\hat{K}}(u)=\mathbb{K}(u)+\mathbb{A}(u)+\mathbb{R}(u),\label{pf:kap}
\end{equation}
where $\mathbb{A}(u)=\frac{1}{2\pi}\int e^{-{\rm i}tu}\frac{K^{{\rm ft}}(t)}{f_{\epsilon}^{{\rm ft}}(t/b_{n})}\hat{\Pi}(t/b_{n})dt$
and $\mathbb{R}(u)=\frac{1}{2\pi}\int e^{-{\rm i}tu}\frac{K^{{\rm ft}}(t)}{f_{\epsilon}^{{\rm ft}}(t/b_{n})}\hat{\Pi}^{{\rm res}}(t/b_{n})dt$
with
\begin{gather*}
\hat{\Pi}(t)=\frac{1}{n}\sum_{l=1}^{n}\Pi_{l}(t),\quad\Pi_{l}(t)=-\frac{\delta_{1,l}(t)}{\mu_{1}(t)}+{\rm i}\int_{0}^{t}\left\{ -\frac{\mu_{3}(s)\delta_{2,l}(s)}{\mu_{2}^{2}(s)}+\frac{\delta_{3,l}(s)}{\mu_{2}(s)}\right\} ds,\\
\hat{\Pi}^{{\rm res}}(t)=\frac{\hat{\delta}_{1}^{2}(t)}{\mu_{1}(t)+\hat{\delta}_{1}(t)}-\int_{0}^{t}{\rm i}\left\{ -\frac{\mu_{3}(s)\hat{\delta}_{2}(s)}{\mu_{2}^{2}(s)}+\frac{\hat{\delta}_{3}(s)}{\mu_{2}(s)}\right\} \frac{\hat{\delta}_{2}(s)}{\mu_{2}(s)+\hat{\delta}_{2}(s)}ds\\
+\int_{0}^{t}{\rm i}\left\{ -\frac{\mu_{3}(s)\hat{\delta}_{2}(s)}{\mu_{2}^{2}(s)}+\frac{\hat{\delta}_{3}(s)}{\mu_{2}(s)}\right\} \left\{ 1-\frac{\hat{\delta}_{2}(s)}{\mu_{2}(s)+\hat{\delta}_{2}(s)}\right\} ds\left\{ -\frac{\hat{\delta}_{1}(t)}{\mu_{1}(t)}+\frac{\hat{\delta}_{1}^{2}(t)}{\mu_{1}(t)+\hat{\delta}_{1}(t)}\right\} \\
-\frac{1}{2}e^{\bar{\phi}(t)}\left(\int_{0}^{t}\left\{ -\frac{\mu_{3}(s)\hat{\delta}_{2}(s)}{\mu_{2}^{2}(s)}+\frac{\hat{\delta}_{3}(s)}{\mu_{2}(s)}\right\} \left\{ 1-\frac{\hat{\delta}_{2}(s)}{\mu_{2}(s)+\hat{\delta}_{2}(s)}\right\} ds\right)^{2}\left\{ 1-\frac{\hat{\delta}_{1}(t)}{\mu_{1}(t)}+\frac{\hat{\delta}_{1}^{2}(t)}{\mu_{1}(t)+\hat{\delta}_{1}(t)}\right\} ,
\end{gather*}
for some $|\bar{\phi}(t)|\le\left|\int_{0}^{t}\left\{ -\frac{\mu_{3}(s)\hat{\delta}_{2}(s)}{\mu_{2}^{2}(s)}+\frac{\hat{\delta}_{3}(s)}{\mu_{2}(s)}\right\} \left\{ 1-\frac{\hat{\delta}_{2}(s)}{\mu_{2}(s)+\hat{\delta}_{2}(s)}\right\} ds\right|$,
where $\hat{\delta}_{\iota}(t)=\frac{1}{n}\sum_{l=1}^{n}\delta_{\iota,l}(t)$
with $\delta_{\iota,l}(t)=\mu_{\iota,l}(t)-\mu_{\iota}(t)$ for $\iota=1,2,3$.
Here, $\mathbb{A}$ denotes the Fréchet derivative of $\hat{\mathbb{K}}$
as a functional of $(\hat{\mu}_{1},\hat{\mu}_{2},\hat{\mu}_{3})$
evaluated at $(\mu_{1},\mu_{2},\mu_{3})$ in the direction of $(\hat{\delta}_{1},\hat{\delta}_{2},\hat{\delta}_{3})$,
and $\mathbb{R}$ contains the remainders. Observe that (\ref{pf:kap}) implies
\begin{align*}
 & \hat{\theta}_{c}=\frac{2}{n^{2}}\sum_{j=1}^{n}\sum_{k=1}^{n}(-1)b_{n}^{-3}(Y_{j}-cW_{j})\int\hat{\mathbb{K}}\left(\frac{x-X_{j}}{b_{n}}\right)\mathbb{\hat{K}}^{\prime}\left(\frac{x-X_{k}}{b_{n}}\right)dx\\
= & \underbrace{\frac{2}{n^{2}}\sum_{j=1}^{n}\sum_{k=1}^{n}(-1)b_{n}^{-3}(Y_{j}-cW_{j})\int\left\{ \begin{array}{c}
\mathbb{K}\left(\frac{x-X_{j}}{b_{n}}\right)\mathbb{K}^{\prime}\left(\frac{x-X_{k}}{b_{n}}\right)\\
+\mathbb{A}\left(\frac{x-X_{j}}{b_{n}}\right)\mathbb{K}^{\prime}\left(\frac{x-X_{k}}{b_{n}}\right)+\mathbb{K}\left(\frac{x-X_{j}}{b_{n}}\right)\mathbb{A}^{\prime}\left(\frac{x-X_{k}}{b_{n}}\right)
\end{array}\right\} dx}_{\eqqcolon S}\\
 & +\underbrace{\frac{2}{n^{2}}\sum_{j=1}^{n}\sum_{k=1}^{n}(-1)b_{n}^{-3}(Y_{j}-cW_{j})\int\left\{ \begin{array}{c}
\mathbb{A}\left(\frac{x-X_{j}}{b_{n}}\right)\mathbb{A}^{\prime}\left(\frac{x-X_{k}}{b_{n}}\right)\\
+\mathbb{R}\left(\frac{x-X_{j}}{b_{n}}\right)\mathbb{K}^{\prime}\left(\frac{x-X_{k}}{b_{n}}\right)+\mathbb{K}\left(\frac{x-X_{j}}{b_{n}}\right)\mathbb{R}^{\prime}\left(\frac{x-X_{k}}{b_{n}}\right)\\
+\mathbb{R}\left(\frac{x-X_{j}}{b_{n}}\right)\mathbb{A}^{\prime}\left(\frac{x-X_{k}}{b_{n}}\right)+\mathbb{A}\left(\frac{x-X_{j}}{b_{n}}\right)\mathbb{R}^{\prime}\left(\frac{x-X_{k}}{b_{n}}\right)\\
+\mathbb{R}\left(\frac{x-X_{j}}{b_{n}}\right)\mathbb{R}^{\prime}\left(\frac{x-X_{k}}{b_{n}}\right)
\end{array}\right\} dx}_{\eqqcolon T}
\end{align*}

First, we are going to show
\begin{equation}
T=o_{p}(n^{-1/2}).\label{pf:t}
\end{equation}
To show (\ref{pf:t}), decompose $T=T_{1}+T_{2}+T_{3}+T_{4}$, where
\begin{align*}
T_{1} & =\frac{-2}{n^{2}b_{n}^{3}}\sum_{j=1}^{n}\sum_{k=1}^{n}(Y_{j}-cW_{j})\int\mathbb{A}\left(\frac{x-X_{j}}{b_{n}}\right)\mathbb{A}^{\prime}\left(\frac{x-X_{k}}{b_{n}}\right)dx,\\
T_{2} & =\frac{-2}{n^{2}b_{n}^{3}}\sum_{j=1}^{n}\sum_{k=1}^{n}(Y_{j}-cW_{j})\int\left\{ \mathbb{R}\left(\frac{x-X_{j}}{b_{n}}\right)\mathbb{K}^{\prime}\left(\frac{x-X_{k}}{b_{n}}\right)+\mathbb{K}\left(\frac{x-X_{j}}{b_{n}}\right)\mathbb{R}^{\prime}\left(\frac{x-X_{k}}{b_{n}}\right)\right\} dx,\\
T_{3} & =\frac{-2}{n^{2}b_{n}^{3}}\sum_{j=1}^{n}\sum_{k=1}^{n}(Y_{j}-cW_{j})\int\left\{ \mathbb{R}\left(\frac{x-X_{j}}{b_{n}}\right)\mathbb{A}^{\prime}\left(\frac{x-X_{k}}{b_{n}}\right)+\mathbb{A}\left(\frac{x-X_{j}}{b_{n}}\right)\mathbb{R}^{\prime}\left(\frac{x-X_{k}}{b_{n}}\right)\right\} dx,\\
T_{4} & =\frac{-2}{n^{2}b_{n}^{3}}\sum_{j=1}^{n}\sum_{k=1}^{n}(Y_{j}-cW_{j})\int\mathbb{R}\left(\frac{x-X_{j}}{b_{n}}\right)\mathbb{R}^{\prime}\left(\frac{x-X_{k}}{b_{n}}\right)dx.
\end{align*}
For $T_{1}$, we have
\begin{eqnarray*}
|T_{1}| & = & \frac{1}{\pi n^{2}b_{n}^{3}}\left|\sum_{j=1}^{n}\sum_{k=1}^{n}(Y_{k}-cW_{k})\iint\left\{ \begin{array}{c}
\frac{1}{2\pi}\int e^{-{\rm i}(t_{1}+t_{2})x/b_{n}}dx\,t_{1}e^{{\rm i}\left(\frac{t_{1}X_{j}+t_{2}X_{k}}{b_{n}}\right)}\\
\times\frac{K^{{\rm ft}}(t_{1})K^{{\rm ft}}(t_{2})}{f_{\epsilon}^{{\rm ft}}(t_{1}/b_{n})f_{\epsilon}^{{\rm ft}}(t_{2}/b_{n})}\hat{\Pi}(t_{1}/b_{n})\hat{\Pi}(t_{2}/b_{n})
\end{array}\right\} dt_{1}dt_{2}\right|\\
 & = & \frac{1}{\pi n^{2}b_{n}^{2}}\left|\sum_{j=1}^{n}\sum_{k=1}^{n}(Y_{k}-cW_{k})\iint\left\{ \begin{array}{c}
\frac{1}{2\pi}\int e^{-{\rm i}(t_{1}+t_{2})\tilde{x}}d\tilde{x}\,t_{1}e^{{\rm i}\left(\frac{t_{1}X_{j}+t_{2}X_{k}}{b_{n}}\right)}\\
\times\frac{K^{{\rm ft}}(t_{1})K^{{\rm ft}}(t_{2})}{f_{\epsilon}^{{\rm ft}}(t_{1}/b_{n})f_{\epsilon}^{{\rm ft}}(t_{2}/b_{n})}\hat{\Pi}(t_{1}/b_{n})\hat{\Pi}(t_{2}/b_{n})
\end{array}\right\} dt_{1}dt_{2}\right|\\
 & = & \frac{{\rm 1}}{\pi n^{2}b_{n}^{2}}\left|\sum_{j=1}^{n}\sum_{k=1}^{n}(Y_{k}-cW_{k})\int te^{{\rm i}t\left(\frac{X_{j}-X_{k}}{b_{n}}\right)}\frac{|K^{{\rm ft}}(t)|^{2}}{|f_{\epsilon}^{{\rm ft}}(t/b_{n})|^{2}}|\hat{\Pi}(t/b_{n})|^{2}dt\right|\\
 & = & O_{p}\left(\frac{\{\sup_{|t|\le b_{n}^{-1}}|\hat{\Pi}(t)|\}^{2}}{b_{n}^{2}\{\inf_{|t|\le b_{n}^{-1}}|f_{\epsilon}^{{\rm ft}}(t)|\}^{2}}\right),
\end{eqnarray*}
where the second step follows from the change of variables $\tilde{x}=x/b_{n}$,
the third step follows by $\int\delta(t-s)f(t)dt=f(s)$ with Dirac
delta function $\delta(t)=\frac{1}{2\pi}\int e^{-{\rm i}tx}dx$, and
the last step uses the implication that $K^{{\rm ft}}$ is supported
on $[-1,1]$ under Assumption (3). Similar arguments show 
\begin{eqnarray*}
|T_{2}| & = & O_{p}\left(\frac{\sup_{|t|\le b_{n}^{-1}}|\hat{\Pi}^{{\rm res}}(t)|}{b_{n}^{2}\{\inf_{|t|\le b_{n}^{-1}}|f_{\epsilon}^{{\rm ft}}(t)|\}^{2}}\right),\\
|T_{3}| & = & O_{p}\left(\frac{\sup_{|t|\le b_{n}^{-1}}|\hat{\Pi}(t)|\sup_{|t|\le b_{n}^{-1}}|\hat{\Pi}^{{\rm res}}(t)|}{b_{n}^{2}\{\inf_{|t|\le b_{n}^{-1}}|f_{\epsilon}^{{\rm ft}}(t)|\}^{2}}\right),\\
|T_{4}| & = & O_{p}\left(\frac{\{\sup_{|t|\le b_{n}^{-1}}|\hat{\Pi}^{{\rm res}}(t)|\}^{2}}{b_{n}^{2}\{\inf_{|t|\le b_{n}^{-1}}|f_{\epsilon}^{{\rm ft}}(t)|\}^{2}}\right),
\end{eqnarray*}
and thus (\ref{pf:t}) follows by Lemma \ref{lem:pi1} and Assumption
(4). 

Hence, for the asymptotic distribution of $\hat{\theta}_{c}$, it is sufficient to focus on $S$. Let $a_{j}=(Y_{j},X_{j},W_{j})$ and let $\text{Sym}(\mathcal{I})$
denote the collection of all permutations of an ordered set $\mathcal{I}$.
Observe that
\begin{align*}
 & S=\frac{2}{n^{2}}\sum_{j=1}^{n}\sum_{k=1}^{n}(-1)b_{n}^{-3}(Y_{j}-W_{j})\int\left\{ \begin{array}{c}
\mathbb{K}\left(\frac{x-X_{j}}{b_{n}}\right)\mathbb{K}^{\prime}\left(\frac{x-X_{k}}{b_{n}}\right)\\
+\mathbb{A}\left(\frac{x-X_{j}}{b_{n}}\right)\mathbb{K}^{\prime}\left(\frac{x-X_{k}}{b_{n}}\right)+\mathbb{K}\left(\frac{x-X_{j}}{b_{n}}\right)\mathbb{A}^{\prime}\left(\frac{x-X_{k}}{b_{n}}\right)
\end{array}\right\} dx\\
= & \frac{2}{n^{3}}\sum_{j=1}^{n}\sum_{k=1}^{n}\sum_{l=1}^{n}\underbrace{(-1)b_{n}^{-3}(Y_{j}-W_{j})\int\left\{ \begin{array}{c}
\mathbb{K}\left(\frac{x-X_{j}}{b_{n}}\right)\mathbb{K}^{\prime}\left(\frac{x-X_{k}}{b_{n}}\right)\\
+\left\{ \frac{1}{2\pi}\int e^{-{\rm i}t\left(\frac{x-X_{j}}{b_{n}}\right)}\frac{K^{{\rm ft}}(t)}{f_{\epsilon}^{{\rm ft}}(t/b_{n})}\Pi_{l}(t/b_{n})dt\right\} \mathbb{K^{\prime}}\left(\frac{x-X_{k}}{b_{n}}\right)\\
+\mathbb{K}\left(\frac{x-X_{j}}{b_{n}}\right)\left\{ \frac{-{\rm i}}{2\pi}\int te^{-{\rm i}t\left(\frac{x-X_{k}}{b_{n}}\right)}\frac{K^{{\rm ft}}(t)}{f_{\epsilon}^{{\rm ft}}(t/b_{n})}\Pi_{l}(t/b_{n})dt\right\} 
\end{array}\right\} dx}_{\eqqcolon q_{n}(a_{j},a_{k},a_{l})}\\
= & \frac{(n-1)(n-2)}{n^{2}}\underbrace{\binom{n}{3}^{-1}\sum_{j=1}^{n-2}\sum_{k=j+1}^{n-1}\sum_{l=k+1}^{n}\overbrace{\sum_{(j',k',l')\in\text{Sym}((j,k,l))}q_{n}(a_{j},a_{k},a_{l})/3}^{\eqqcolon p_{n}(a_{j},a_{k},a_{l})}-E[p_{n}(a_{j},a_{k},a_{l})]}_{\eqqcolon U_{n}}\\
 & +\underbrace{\frac{(n-1)(n-2)}{n^{2}}E[p_{n}(a_{1},a_{2},a_{3})]}_{\eqqcolon B_{n}}\\
 & +\underbrace{\frac{2}{n^{3}}\left\{ \sum_{j=1}^{n-1}\sum_{k=j+1}^{n}\left\{ \begin{array}{c}
q_{n}(a_{j},a_{j},a_{k})+q_{n}(a_{k},a_{k},a_{j})\\
+q_{n}(a_{j},a_{k},a_{j})+q_{n}(a_{k},a_{j},a_{k})\\
+q_{n}(a_{j},a_{k},a_{k})+q_{n}(a_{k},a_{j},a_{j})
\end{array}\right\} +\sum_{j=1}^{n}q_{n}(a_{j},a_{j},a_{j})\right\} }_{\eqqcolon R_{n}},
\end{align*}
where $U_{n}$ is a 3rd order U-statistic with symmetric kernel $p_{n}$, 
and its Hájek projection is given by
\[
U_{n}^{*}=\frac{3}{n}\sum_{j=1}^{n}\{r_{n}(a_{j})-E[r_{n}(a_{j})]\},
\]
where $r_{n}(a_{j})=E_{j}[p_{n}(a_{j},a_{k},a_{l})]$ and $E_{j}[\cdot]=E[\cdot|a_{j}]$.
Then, we can write
\begin{equation}
S-\theta_{c}=U_{n}^{*}(1-1/n)(1-2/n)+\{U_{n}-U_{n}^{*}\}(1-1/n)(1-2/n)+\{B_{n}-\theta_{c}\}+R_{n}.\label{pf:dec2}
\end{equation}
First, we are going to show
\begin{equation}
R_{n}=o_{p}(n^{-1/2}).\label{pf:r}
\end{equation}
To show (\ref{pf:r}), decompose $R_{n}=R_{n,1}+R_{n,2}+R_{n,3}$,
where
\begin{align*}
	R_{n,1} &=\frac{2}{n^{3}}\sum_{j=1}^{n-1}\sum_{k=j+1}^{n}\{q_{n}(a_{j},a_{j},a_{k})+q_{n}(a_{k},a_{k},a_{j})+q_{n}(a_{j},a_{k},a_{k})+q_{n}(a_{k},a_{j},a_{j})\}
	\\
	R_{n,2}  &=\frac{2}{n^{3}}\sum_{j=1}^{n-1}\sum_{k=j+1}^{n}\{q_{n}(a_{j},a_{k},a_{j})+q_{n}(a_{k},a_{j},a_{k})\},
	\\
	R_{n,3}
	&=\frac{2}{n^{3}}\sum_{j=1}^{n}q_{n}(a_{j},a_{j},a_{j}).
\end{align*}

Observe that
\begin{align*}
q_{n}(d_{j},d_{k},d_{l}) & =\frac{{\rm i}(Y_{j}-cW_{j})}{2\pi b_{n}^{3}}\iint\left\{ \begin{array}{c}
\left\{ \frac{1}{2\pi}\int e^{-{\rm i}(t_{1}+t_{2})x/b_{n}}dx\right\} t_{2}e^{{\rm i}\left(\frac{t_{1}X_{j}+t_{2}X_{k}}{b_{n}}\right)}\\
\times\frac{K^{\mathrm{ft}}(t_{1})K^{\mathrm{ft}}(t_{2})}{f_{\epsilon}^{\mathrm{ft}}(t_{1}/b_{n})f_{\epsilon}^{\mathrm{ft}}(t_{2}/b_{n})}\{1+\Pi_{l}(t_{1}/b_{n})+\Pi_{l}(t_{2}/b_{n})\}
\end{array}\right\} dt_{1}dt_{2}\\
 & =\frac{{\rm i}(Y_{j}-cW_{j})}{2\pi b_{n}^{2}}\iint\left\{ \begin{array}{c}
\left\{ \frac{1}{2\pi}\int e^{-{\rm i}(t_{1}+t_{2})\tilde{x}}d\tilde{x}\right\} t_{2}e^{{\rm i}\left(\frac{t_{1}X_{j}+t_{2}X_{k}}{b_{n}}\right)}\\
\times\frac{K^{\mathrm{ft}}(t_{1})K^{\mathrm{ft}}(t_{2})}{f_{\epsilon}^{\mathrm{ft}}(t_{1}/b_{n})f_{\epsilon}^{\mathrm{ft}}(t_{2}/b_{n})}\{1+\Pi_{l}(t_{1}/b_{n})+\Pi_{l}(t_{2}/b_{n})\}
\end{array}\right\} dt_{1}dt_{2}\\
 & =\frac{{\rm i}(cW_{j}-Y_{j})}{2\pi b_{n}^{2}}\int te^{{\rm i}t\left(\frac{X_{j}-X_{k}}{b_{n}}\right)}\frac{|K^{\mathrm{ft}}(t)|^{2}}{|f_{\epsilon}^{\mathrm{ft}}(t/b_{n})|^{2}}\left\{ 1+\Pi_{l}(t/b_{n})+\Pi_{l}(-t/b_{n})\right\} dt,
\end{align*}
where the third step follows from the change of variable $\tilde{x}=x/b_{n}$
and the last step follows from the property of Dirac delta function.
Using the fact that $K^{{\rm ft}}$ is supported on $[-1,1]$ under
Assumption (3), this implies
\begin{align*}
E[|R_{n,1}|] & \le\frac{2(n-1)}{n^{2}}\{E[|q_{n}(a_{1},a_{1},a_{2})|]+E[|q_{n}(a_{1},a_{2},a_{2})|]\}=O\left(\frac{\max\{1,\sup_{|t|\le b_{n}^{-1}}E[|\Pi_{1}(t)|]\}}{nb_{n}^{2}\{\inf_{|t|\le b_{n}^{-1}}|f_{\epsilon}^{{\rm ft}}(t)|\}^{2}}\right),\\
E[|R_{n,2}|] & \le\frac{2(n-1)}{n^{2}}E[|q_{n}(a_{1},a_{2},a_{1})|]=O\left(\frac{\max\{1,\sup_{|t|\le b_{n}^{-1}}E[|(Y_{1}-cW_{1})\Pi_{1}(t)|]\}}{nb_{n}^{2}\{\inf_{|t|\le b_{n}^{-1}}|f_{\epsilon}^{{\rm ft}}(t)|\}^{2}}\right),\\
E[|R_{n,3}|] & \le\frac{2}{n^{2}}E[|q_{n}(a_{1},a_{1},a_{1})|]=O\left(\frac{\max\{1,\sup_{|t|\le b_{n}^{-1}}E[|(Y_{1}-cW_{1})\Pi_{1}(t)|]\}}{n^{2}b_{n}^{2}\{\inf_{|t|\le b_{n}^{-1}}|f_{\epsilon}^{{\rm ft}}(t)|\}^{2}}\right),
\end{align*}
and (\ref{pf:r}) follows by Lemma \ref{lem:pi2} and Assumption (4).

Second, under Assumption (4), we have
\begin{equation}
B_{n}-\theta_{c}=o(n^{-1/2}),\label{pf:b}
\end{equation}
which follows from
\begin{align*}
 & E[p_{n}(a_{1},a_{2},a_{3})]=2E[q_{n}(a_{1},a_{2},a_{3})]
 \\
 =&-\frac{2}{b_{n}^{3}}\int E\left[(Y-cW)\mathbb{K}\left(\frac{x-X}{b_{n}}\right)\right]E\left[\mathbb{K}^{\prime}\left(\frac{x-X}{b_{n}}\right)\right]dx\\
= & -2\int\left\{ \frac{1}{b_{n}}E\left[(Y-cW)K\left(\frac{x-X^{*}}{b_{n}}\right)\right]\right\} \left\{ \frac{1}{b_{n}^{2}}E\left[K^{\prime}\left(\frac{x-X^{*}}{b_{n}}\right)\right]\right\} dx\\
= & \underbrace{-2\int h_{c}(x)f^{\prime}(x)dx}_{\theta_{c}}-2\int\left\{ \begin{array}{c}
h(x)\frac{b_{n}^{\alpha-1}}{(\alpha-1)!}\int K(u)\{\overbrace{f^{(\alpha)}(x+b_{n}\bar{u}_{f})-f^{(\alpha)}(x)}^{\le m(x)b_{n}|u|}\}u^{\alpha-1}du\\
+f^{\prime}(x)\frac{b_{n}^{\alpha}}{\alpha!}\int K(u)\{\underbrace{h_{c}^{(\alpha)}(x+b_{n}\bar{u}_{h})-h_{c}^{(\alpha)}(x)}_{\le m(x)b_{n}|u|}\}u^{\alpha}du\\
+\left\{ \begin{array}{c}
\frac{b_{n}^{\alpha-1}}{(\alpha-1)!}\int K(u)\{f^{(\alpha)}(x+b_{n}\bar{u}_{f})-f^{(\alpha)}(x)\}u^{\alpha-1}du\\
\times\frac{b_{n}^{\alpha}}{\alpha!}\int K(u)\{h_{c}^{(\alpha)}(x+b_{n}\bar{u}_{h})-h_{c}^{(\alpha)}(x)\}u^{\alpha}du
\end{array}\right\} 
\end{array}\right\} dx\\
= & \theta_{c}+O(b_{n}^{\alpha})
\end{align*}
for some $\bar{u}_{h}$ and $\bar{u}_{f}$ such that $\max\{|\bar{u}_{h}|,|\bar{u}_{f}|\}\le|u|$,
where the second step follows from $E[\Pi_{l}(t/b_{n})]=0$, the third
step follows from Lemma \ref{lem:bias1}, the fourth step follows
from Lemma \ref{lem:bias2}, and the last step follows from the Lipschitz
conditions on $h_{c}^{(\alpha)}$ and $f^{(\alpha)}$ under Assumption
(2). 

Also, note by Lemma A.3 of Ahn and Powell (1993) that
\begin{equation}
U_{n}-U_{n}^{*}=o_{p}(n^{-1/2})\label{pf:u}
\end{equation}
if $E[|p_{n}(a_{j},a_{k},a_{l})|^{2}]=o(n)$, which follows from Assumption
(4), Lemma \ref{lem:pi2}, and
\begin{align*}
 & E[|p_{n}(a_{j},a_{k},a_{l})|^{2}]\le4E[|q_{n}(a_{j},a_{k},a_{l})|^{2}]\\
= & 4E\left[\left|\frac{{\rm i}(cW_{j}-Y_{j})}{2\pi b_{n}^{2}}\int te^{{\rm i}t\left(\frac{X_{j}-X_{k}}{b_{n}}\right)}\frac{|K^{\mathrm{ft}}(t)|^{2}}{|f_{\epsilon}^{\mathrm{ft}}(t/b_{n})|^{2}}\left\{ 1+\Pi_{l}(t/b_{n})+\Pi_{l}(-t/b_{n})\right\} dt\right|^{2}\right]\\
\le & \frac{E[|Y_{j}-cW_{j}|^{2}]}{b_{n}^{4}}\iint\left\{ \begin{array}{c}
|t_{1}t_{2}|\frac{|K^{\mathrm{ft}}(t_{1})|^{2}|K^{\mathrm{ft}}(t_{2})|^{2}}{|f_{\epsilon}^{\mathrm{ft}}(t_{1}/b_{n})|^{2}|f_{\epsilon}^{\mathrm{ft}}(t_{2}/b_{n})|^{2}}\\
\times E\left[\begin{array}{c}
\left\{ 1+|\Pi_{l}(t_{1}/b_{n})|+|\Pi_{l}(-t_{1}/b_{n})|\right\} \\
\times\left\{ 1+|\Pi_{l}(t_{2}/b_{n})|+|\Pi_{l}(-t_{2}/b_{n})|\right\} 
\end{array}\right]
\end{array}\right\} dt_{1}dt_{2}\\
= & O\left(\frac{\max\{1,\sup_{|t|\le b_{n}^{-1}}E[|\Pi_{1}(t)|^{2}]\}}{b_{n}^{4}\{\inf_{|t|\le b_{n}^{-1}}|f_{\epsilon}^{{\rm ft}}(t)|\}^{4}}\right),
\end{align*}
where the second step follows by $k\neq l$ and the last step follows
by the fact that $K^{{\rm ft}}$ is supported on $[-1,1]$ as in Assumption
(3) and the Cauchy-Schwarz inequality.

Finally, observe that
\begin{align*}
 & 3r_{n}(a_{j})=E_{j}\left[\begin{array}{c}
q_{n}(a_{j},a_{k},a_{l})+q_{n}(a_{j},a_{l},a_{k})\\
+q_{n}(a_{k},a_{j},a_{l})+q_{n}(a_{l},a_{j},a_{k})\\
+q_{n}(a_{k},a_{l},a_{j})+q_{n}(a_{l},a_{k},a_{j})
\end{array}\right]=2\left\{ \begin{array}{c}
E_{j}[q_{n}(a_{j},a_{k},a_{l})]\\
+E_{j}[q_{n}(a_{k},a_{j},a_{l})]\\
+E_{j}[q_{n}(a_{k},a_{l},a_{j})]
\end{array}\right\} \\
= & \underbrace{(-2)b_{n}^{-3}\int\left\{ \begin{array}{c}
E\left[(Y-cW)\mathbb{K}\left(\frac{x-X}{b_{n}}\right)\right]E\left[\mathbb{K}^{\prime}\left(\frac{x-X}{b_{n}}\right)\right]\\
-\left\{ \frac{1}{2\pi}\int e^{-{\rm i}tx/b_{n}}h_{c}^{{\rm ft}}(t/b_{n})E[\Pi_{j}^{*}(t/b_{n})]K^{{\rm ft}}(t)dt\right\} E\left[\mathbb{K}^{\prime}\left(\frac{x-X}{b_{n}}\right)\right]\\
-E\left[(Y-cW)\mathbb{K}\left(\frac{x-X}{b_{n}}\right)\right]\left\{ \frac{-{\rm i}}{2\pi}\int te^{-{\rm i}tx/b_{n}}f^{{\rm ft}}(t/b_{n})E[\Pi_{j}^{*}(t/b_{n})]K^{{\rm ft}}(t)dt\right\} 
\end{array}\right\} dx}_{\eqqcolon c_{r}}\\
 & +\underbrace{(-2)b_{n}^{-3}\int\left\{ \begin{array}{c}
(Y_{j}-cW_{j})\mathbb{K}\left(\frac{x-X_{j}}{b_{n}}\right)E\left[\mathbb{K}^{\prime}\left(\frac{x-X}{b_{n}}\right)\right]+E\left[(Y-cW)\mathbb{K}\left(\frac{x-X}{b_{n}}\right)\right]\mathbb{K}^{\prime}\left(\frac{x-X_{j}}{b_{n}}\right)\\
+\left\{ \frac{1}{2\pi}\int e^{-{\rm i}tx/b_{n}}h_{c}^{{\rm ft}}(t/b_{n})\Pi_{j}^{*}(t/b_{n})K^{{\rm ft}}(t)dt\right\} E\left[\mathbb{K}^{\prime}\left(\frac{x-X}{b_{n}}\right)\right]\\
+E\left[(Y-cW)\mathbb{K}\left(\frac{x-X}{b_{n}}\right)\right]\left\{ \frac{-{\rm i}}{2\pi}\int te^{-{\rm i}tx/b_{n}}f^{{\rm ft}}(t/b_{n})\Pi_{j}^{*}(t/b_{n})K^{{\rm ft}}(t)dt\right\} 
\end{array}\right\} dx}_{\eqqcolon r_{n}^{*}(a_{j})},
\end{align*}
where the last step follows from $\Pi_{j}(t/b_{n})=\Pi_{j}^{*}(t)-E[\Pi_{j}^{*}(t)]$
(so $E[\Pi_{j}(t/b_{n})]=0$) with $\Pi_{j}^{*}(t)=-\frac{\mu_{1,j}(t)}{\mu_{1}(t)}+{\rm i}\int_{0}^{t}\left\{ -\frac{\mu_{3}(s)\mu_{2,j}(s)}{\mu_{2}^{2}(s)}+\frac{\mu_{3,j}(s)}{\mu_{2}(s)}\right\} ds$.

Since $c_{r}$ is non-stochastic, to characterize the behavior of
$U_{n}^{*}$, it is sufficient to focus on $r_{n}^{*}(a_{j})$, for
which we have
\begin{align*}
 & r_{n}^{*}(a_{j})=2b_{n}^{-3}\int\left\{ \begin{array}{c}
\int\mathbb{K}\left(\frac{x-X_{j}}{b_{n}}\right)E\left[(Y-cW)\mathbb{K}^{\prime}\left(\frac{x-X}{b_{n}}\right)\right]\\
-(Y_{j}-cW_{j})\mathbb{K}\left(\frac{x-X_{j}}{b_{n}}\right)E\left[\mathbb{K}^{\prime}\left(\frac{x-X}{b_{n}}\right)\right]\\
-\left\{ \frac{1}{2\pi}\int e^{-{\rm i}tx/b_{n}}f^{{\rm ft}}(t/b_{n})\Pi_{j}^{*}(t/b_{n})K^{{\rm ft}}(t)dt\right\} E\left[(Y-cW)\mathbb{K}^{\prime}\left(\frac{x-X}{b_{n}}\right)\right]\\
+\left\{ \frac{1}{2\pi}\int e^{-{\rm i}tx/b_{n}}h_{c}^{{\rm ft}}(t/b_{n})\Pi_{j}^{*}(t/b_{n})K^{{\rm ft}}(t)dt\right\} E\left[\mathbb{K}^{\prime}\left(\frac{x-X}{b_{n}}\right)\right]
\end{array}\right\} dx\\
= & 2b_{n}^{-1}\int\left\{ \begin{array}{c}
\int\mathbb{K}\left(\frac{x-X_{j}}{b_{n}}\right)\left\{ b_{n}^{-2}E\left[(Y-cW)\mathbb{K}^{\prime}\left(\frac{x-X}{b_{n}}\right)\right]\right\} \\
-(Y_{j}-cW_{j})\mathbb{K}\left(\frac{x-X_{j}}{b_{n}}\right)\left\{ b_{n}^{-2}E\left[\mathbb{K}^{\prime}\left(\frac{x-X}{b_{n}}\right)\right]\right\} \\
-\left\{ \frac{1}{2\pi}\int e^{-{\rm i}tx/b_{n}}f^{{\rm ft}}(t/b_{n})\Pi_{j}^{*}(t/b_{n})K^{{\rm ft}}(t)dt\right\} \left\{ b_{n}^{-2}E\left[(Y-cW)\mathbb{K}^{\prime}\left(\frac{x-X}{b_{n}}\right)\right]\right\} \\
+\left\{ \frac{1}{2\pi}\int e^{-{\rm i}tx/b_{n}}h_{c}^{{\rm ft}}(t/b_{n})\Pi_{j}^{*}(t/b_{n})K^{{\rm ft}}(t)dt\right\} \left\{ b_{n}^{-2}E\left[\mathbb{K}^{\prime}\left(\frac{x-X}{b_{n}}\right)\right]\right\} 
\end{array}\right\} dx\\
= & \frac{1}{\pi}\int\left\{ \begin{array}{c}
\int\left\{ \{h_{c}^{\prime}\}^{{\rm ft}}(-t)-(Y_{j}-cW_{j})\{f^{\prime}\}^{{\rm ft}}(-t)\right\} \frac{e^{{\rm i}tX_{j}}}{f_{\epsilon}^{{\rm ft}}(t)}\\
+\left\{ f^{{\rm ft}}(t)\{h_{c}^{\prime}\}^{{\rm ft}}(-t)-\{f^{\prime}\}^{{\rm ft}}(-t)h_{c}^{{\rm ft}}(t)\right\} \Pi_{j}^{*}(t)
\end{array}\right\} dt +v_{n,1}(a_{j})+v_{n,2}(a_{j}),
\end{align*}
where the first step uses the integration by parts and $v_{n,1}(a_{j})$
and $v_{n,2}(a_{j})$ are defined as
\begin{gather*}
v_{n,1}(a_{j})=2b_{n}^{-1}\int\left\{ \begin{array}{c}
\left\{ \begin{array}{c}
\frac{1}{2\pi}\int e^{-{\rm i}tx/b_{n}}\left[\frac{e^{{\rm i}tX_{j}/b_{n}}}{f_{\epsilon}^{{\rm ft}}(t/b_{n})}+f^{{\rm ft}}(t/b_{n})\Pi_{j}^{*}(t/b_{n})\right]K^{{\rm ft}}(t)dt\\
\times\left\{ b_{n}^{-2}E\left[(Y-cW)\mathbb{K}^{\prime}\left(\frac{x-X}{b_{n}}\right)\right]-h_{c}^{\prime}(x)\right\} 
\end{array}\right\} \\
-\left\{ \begin{array}{c}
\frac{1}{2\pi}\int e^{-{\rm i}tx/b_{n}}\left[\frac{(Y_{j}-cW_{j})e^{{\rm i}tX_{j}/b_{n}}}{f_{\epsilon}^{{\rm ft}}(t/b_{n})}+h_{c}^{{\rm ft}}(t/b_{n})\Pi_{j}^{*}(t/b_{n})\right]K^{{\rm ft}}(t)dt\\
\times\left\{ b_{n}^{-2}E\left[\mathbb{K}^{\prime}\left(\frac{x-X}{b_{n}}\right)\right]-f^{\prime}(x)\right\} 
\end{array}\right\} 
\end{array}\right\} dx,\\
v_{n,2}(a_{j})=\frac{1}{\pi}\int\left\{ \begin{array}{c}
\int\left\{ \{h_{c}^{\prime}\}^{{\rm ft}}(-t)-(Y_{j}-cW_{j})\{f^{\prime}\}^{{\rm ft}}(-t)\right\} \frac{e^{{\rm i}tX_{j}}}{f_{\epsilon}^{{\rm ft}}(t)}\\
+\left\{ f^{{\rm ft}}(t)\{h_{c}^{\prime}\}^{{\rm ft}}(-t)-\{f^{\prime}\}^{{\rm ft}}(-t)h_{c}^{{\rm ft}}(t)\right\} \Pi_{j}^{*}(t)
\end{array}\right\} \left\{ K^{{\rm ft}}(tb_{n})-1\right\} dt.
\end{gather*}
Since $Var[\xi_{c,j}]<\infty$ under Assumption (5), $Var[v_{n,2}(a_{j})]=o(1)$
as $K^{{\rm ft}}(tb_{n})\to1$ as $n\to\infty$, and the conclusion
follows if 
\begin{align}
Var[v_{n,1}(a_{j})] & =o(1).\label{pf:v}
\end{align}
To show (\ref{pf:v}), using Lemma \ref{lem:bias1} and \ref{lem:bias2},
we can write $v_{n,1}(a_{j})=v_{n,1,1}(a_{j})+v_{n,1,2}(a_{j})$,
where
\begin{align*}
v_{n,1,1}(a_{j}) & =\frac{b_{n}^{\alpha-2}}{\pi(\alpha-1)!}\int\left\{ \begin{array}{c}
	\int e^{-{\rm i}tx/b_{n}}\left\{ {\rm i}f^{{\rm ft}}(t/b_{n})\int_{0}^{t/b_{n}}\left\{ -\frac{\{f^{{\rm ft}}\}^{\prime}(s)}{f^{{\rm ft}}(s)}+{\rm i}X_{j}\right\} \frac{e^{{\rm i}sW_{j}}}{f^{{\rm ft}}(s)f_{\nu}^{{\rm ft}}(s)}ds\right\} K^{{\rm ft}}(t)dt\\
	\times\int K(u)\{h_{c}^{(\alpha)}(x+b_{n}\bar{u}_{h})-h_{c}^{(\alpha)}(x)\}u^{\alpha-1}du
\end{array}\right\} dx,
\\
v_{n,1,2}(a_{j}) & =\frac{-b_{n}^{\alpha-2}}{\pi(\alpha-1)!}\int\left\{ \begin{array}{c}
\int e^{-{\rm i}tx/b_{n}}\left\{ \begin{array}{c}
\left\{ (Y_{j}-cW_{j})-\frac{h_{c}^{{\rm ft}}(t/b_{n})}{f^{{\rm ft}}(t/b_{n})}\right\} \frac{e^{{\rm i}tX_{j}/b_{n}}}{f_{\epsilon}^{{\rm ft}}(t/b_{n})}\\
+{\rm i}h_{c}^{{\rm ft}}(t/b_{n})\int_{0}^{t/b_{n}}\left\{ -\frac{\{f^{{\rm ft}}\}^{\prime}(s)}{f^{{\rm ft}}(s)}+{\rm i}X_{j}\right\} \frac{e^{{\rm i}sW_{j}}}{f^{{\rm ft}}(s)f_{\nu}^{{\rm ft}}(s)}ds
\end{array}\right\} K^{{\rm ft}}(t)dt\\
\times\int K(u)\{f^{(\alpha)}(x+b_{n}\bar{u}_{f})-f^{(\alpha)}(x)\}u^{\alpha-1}du
\end{array}\right\} dx.
\end{align*}

For $v_{n,1,1}(a_{j})$, we have
\begin{align*}
 & Var[v_{n,1,1}(a_{j})]\le E[|v_{n,1,1}(a_{j})|^{2}]\\
= & E\left[\left|\frac{b_{n}^{\alpha-1}}{\pi(\alpha-1)!}\int\left\{ \begin{array}{c}
\int e^{-{\rm i}\tilde{t}x}\left\{ {\rm i}f^{{\rm ft}}(\tilde{t})\int_{0}^{\tilde{t}}\left\{ -\frac{\{f^{{\rm ft}}\}^{\prime}(s)}{f^{{\rm ft}}(s)}+{\rm i}X_{j}\right\} \frac{e^{{\rm i}sW_{j}}}{f^{{\rm ft}}(s)f_{\nu}^{{\rm ft}}(s)}ds\right\} K^{{\rm ft}}(\tilde{t}b_{n})d\tilde{t}\\
\times\int K(u)\{h_{c}^{(\alpha)}(x+b_{n}\bar{u}_{h})-h_{c}^{(\alpha)}(x)\}u^{\alpha-1}du
\end{array}\right\} dx\right|^{2}\right]\\
\le & E\left[\left\{ b_{n}^{\alpha-1}\int\left\{ \begin{array}{c}
\int\left\{ \int_{0}^{\tilde{t}}\left\{ \frac{|\{f^{{\rm ft}}\}^{\prime}(s)|}{|f^{{\rm ft}}(s)|}+|X_{j}|\right\} \frac{1}{|f^{{\rm ft}}(s)||f_{\nu}^{{\rm ft}}(s)|}ds\right\} |K^{{\rm ft}}(\tilde{t}b_{n})|d\tilde{t}\\
\times\int K(u)|\underbrace{h_{c}^{(\alpha)}(x+b_{n}\bar{u}_{h})-h_{c}^{(\alpha)}(x)}_{\le m(x)b_{n}|u|}||u|^{\alpha-1}du
\end{array}\right\} dx\right\} ^{2}\right]\\
= & O\left(\frac{b_{n}^{2(\alpha-1)}}{\{\inf_{|t|\le b_{n}^{-1}}|f_{\nu}^{{\rm ft}}(t)|\}^{2}\{\inf_{|t|\le b_{n}^{-1}}|f^{{\rm ft}}(t)|\}^{4}}\right),
\end{align*}
where the second step follows from the change of variables $\tilde{t}=t/b_{n}$
and the last step uses the fact that $K^{{\rm ft}}$ is supported
on $[-1,1]$ as in Assumption (3). By similar argument, we can show
\[
Var[v_{n,1,2}(a_{j})]=O\left(\frac{b_{n}^{2(\alpha-1)}\{\inf_{|t|\le b_{n}^{-1}}|f^{{\rm ft}}(t)|\}^{-2}}{\min\left\{ \{\inf_{|t|\le b_{n}^{-1}}|f_{\epsilon}^{{\rm ft}}(t)|\}^{2},\{\inf_{|t|\le b_{n}^{-1}}|f_{\nu}^{{\rm ft}}(t)|\}^{2}\{\inf_{|t|\le b_{n}^{-1}}|f^{{\rm ft}}(t)|\}^{2}\right\} }\right),
\]
and (\ref{pf:v}) follows from Assumption (4).

\newpage{}

\section{Lemmas }
\begin{lem}
\label{lem:delta} Under Assumption (1), for $\iota=1,2,3$,
\begin{align*}
\sup_{|t|\le b_{n}^{-1}}|\hat{\delta}_{\iota}(t)| & =O_{p}\left(n^{-1/2}\log(1/b_{n})\right).
\end{align*}
\end{lem}

\begin{proof}
See Lemma 2 in Kurisu and Otsu (2022).
\end{proof}
\begin{lem}
\label{lem:pi1} Under Assumptions (1) and (4), 
\begin{gather*}
\sup_{|t|\le b_{n}^{-1}}|\hat{\Pi}(t)|=O_{p}\left(\frac{n^{-1/2}\log(1/b_{n})\{\inf_{|t|\le b_{n}^{-1}}|f^{{\rm ft}}(t)|\}^{-1}}{\min\left\{ \inf_{|t|\le b_{n}^{-1}}|f_{\epsilon}^{{\rm ft}}(t)|,\{\inf_{|t|\le b_{n}^{-1}}|f_{\nu}^{{\rm ft}}(t)|\}^{2}\inf_{|t|\le b_{n}^{-1}}|f^{{\rm ft}}(t)|b_{n}\right\} }\right),\\
\sup_{|t|\le b_{n}^{-1}}|\hat{\Pi}^{{\rm res}}(t)|=O_{p}\left(\frac{n^{-1}\log(1/b_{n})^{2}\{\inf_{|t|\le b_{n}^{-1}}|f^{{\rm ft}}(t)|\}^{-1}}{\min\left\{ \inf_{|t|\le b_{n}^{-1}}|f_{\epsilon}^{{\rm ft}}(t)|,\{\inf_{|t|\le b_{n}^{-1}}|f_{\nu}^{{\rm ft}}(t)|\}^{4}\{\inf_{|t|\le b_{n}^{-1}}|f^{{\rm ft}}(t)|\}^{3}b_{n}^{2}\right\} }\right).
\end{gather*}
\end{lem}

\begin{proof}
The first statement follows by
\begin{eqnarray*}
  \sup_{|t|\le b_{n}^{-1}}|\hat{\Pi}(t)|
 & = & O_{p}\left(\frac{\sup_{|t|\le b_{n}^{-1}}|\hat{\delta}_{1}(t)|}{\inf_{|t|\le b_{n}^{-1}}|\mu_{1}(t)|}+b_{n}^{-1}\left\{ \frac{\sup_{|t|\le b_{n}^{-1}}|\mu_{3}(t)|\sup_{|t|\le b_{n}^{-1}}|\hat{\delta}_{2}(t)|}{\{\inf_{|t|\le b_{n}^{-1}}|\mu_{2}(t)|\}^{2}}+\frac{\sup_{|t|\le b_{n}^{-1}}|\hat{\delta}_{3}(t)|}{\inf_{|t|\le b_{n}^{-1}}|\mu_{2}(t)|}\right\} \right)\\
 & = & O_{p}\left(\frac{n^{-1/2}\log(1/b_{n})\{\inf_{|t|\le b_{n}^{-1}}|f^{{\rm ft}}(t)|\}^{-1}}{\min\left\{ \inf_{|t|\le b_{n}^{-1}}|f_{\epsilon}^{{\rm ft}}(t)|,\{\inf_{|t|\le b_{n}^{-1}}|f_{\nu}^{{\rm ft}}(t)|\}^{2}\inf_{|t|\le b_{n}^{-1}}|f^{{\rm ft}}(t)|b_{n}\right\} }\right),
\end{eqnarray*}
where the last step uses Lemma \ref{lem:delta}, $\frac{\inf_{|t|\le b_{n}^{-1}}|\mu_{1}(t)|}{\inf_{|t|\le b_{n}^{-1}}|f_{\epsilon}^{{\rm ft}}(t)|\inf_{|t|\le b_{n}^{-1}}|f^{{\rm ft}}(t)|}\ge1$,
$\frac{\inf_{|t|\le b_{n}^{-1}}|\mu_{2}(t)|}{\inf_{|t|\le b_{n}^{-1}}|f_{\nu}^{{\rm ft}}(t)|\inf_{|t|\le b_{n}^{-1}}|f^{{\rm ft}}(t)|}\ge1$,
and $\sup_{|t|\le b_{n}^{-1}}|\mu_{3}(t)|=O(1)$ under Assumption
(1).

For the second statement, observe that,
\begin{align*}
 \sup_{|t|\le b_{n}^{-1}}|\bar{\phi}(t)|
= & O_{p}\left(\begin{array}{c}
b_{n}^{-1}\left\{ \frac{\sup_{|t|\le b_{n}^{-1}}|\mu_{3}(t)|\sup_{|t|\le b_{n}^{-1}}|\hat{\delta}_{2}(t)|}{\{\inf_{|t|\le b_{n}^{-1}}|\mu_{2}(t)|\}^{2}}+\frac{\sup_{|t|\le b_{n}^{-1}}|\hat{\delta}_{3}(t)|}{\inf_{|t|\le b_{n}^{-1}}|\mu_{2}(t)|}\right\} \\
\times\left\{ 1+\frac{\sup_{|t|\le b_{n}^{-1}}|\hat{\delta}_{2}(t)|}{\inf_{|t|\le b_{n}^{-1}}|\mu_{2}(t)+\hat{\delta}_{2}(t)|}\right\} 
\end{array}\right),\\
= & O_{p}\left(\frac{n^{-1/2}b_{n}^{-1}\log(1/b_{n})}{\{\inf_{|t|\le b_{n}^{-1}}|f_{\nu}^{{\rm ft}}(t)|\}^{2}\{\inf_{|t|\le b_{n}^{-1}}|f^{{\rm ft}}(t)|\}^{2}}\right)=o_{p}(1),
\end{align*}
where the second step uses Lemma \ref{lem:delta}, $\frac{\inf_{|t|\le b_{n}^{-1}}|\mu_{2}(t)|}{\inf_{|t|\le b_{n}^{-1}}|f_{\nu}^{{\rm ft}}(t)|\inf_{|t|\le b_{n}^{-1}}|f^{{\rm ft}}(t)|}\ge1$,
and $\sup_{|t|\le b_{n}^{-1}}|\mu_{3}(t)|=O(1)$ under Assumption
(1), and the last step follows from Assumption (4), which implies
$\sup_{|t|\le b_{n}^{-1}}e^{|\bar{\phi}(t)|}=O_{p}(1)$. The conclusion
then follows by
\begin{align*}
\sup_{|t|\le b_{n}^{-1}}|\hat{\Pi}^{{\rm res}}(t)| & =O_{p}\left(\begin{array}{c}
\frac{\{\sup_{|t|\le b_{n}^{-1}}|\hat{\delta}_{1}(t)|\}^{2}}{\inf_{|t|\le b_{n}^{-1}}|\mu_{1}(t)+\hat{\delta}_{1}(t)|}+b_{n}^{-1}\left\{ \frac{\sup_{|t|\le b_{n}^{-1}}|\mu_{3}(t)|\sup_{|t|\le b_{n}^{-1}}|\hat{\delta}_{2}(t)|}{\{\inf_{|t|\le b_{n}^{-1}}|\mu_{2}(t)|\}^{2}}+\frac{\sup_{|t|\le b_{n}^{-1}}|\hat{\delta}_{3}(t)|}{\inf_{|t|\le b_{n}^{-1}}|\mu_{2}(t)|}\right\} \\
\times\left\{ \begin{array}{c}
\frac{\sup_{|t|\le b_{n}^{-1}}|\hat{\delta}_{1}(t)|}{\inf_{|t|\le b_{n}^{-1}}|\mu_{1}(t)|}+\frac{\{\sup_{|t|\le b_{n}^{-1}}|\hat{\delta}_{1}(t)|\}^{2}}{\inf_{|t|\le b_{n}^{-1}}|\mu_{1}(t)+\hat{\delta}_{1}(t)|}+\frac{\sup_{|t|\le b_{n}^{-1}}|\hat{\delta}_{2}(t)|}{\inf_{|t|\le b_{n}^{-1}}|\mu_{2}(t)+\hat{\delta}_{2}(t)|}\\
+\frac{\sup_{|t|\le b_{n}^{-1}}|\hat{\delta}_{2}(t)|}{\inf_{|t|\le b_{n}^{-1}}|\mu_{2}(t)+\hat{\delta}_{2}(t)|}\left\{ \frac{\sup_{|t|\le b_{n}^{-1}}|\hat{\delta}_{1}(t)|}{\inf_{|t|\le b_{n}^{-1}}|\mu_{1}(t)|}+\frac{\{\sup_{|t|\le b_{n}^{-1}}|\hat{\delta}_{1}(t)|\}^{2}}{\inf_{|t|\le b_{n}^{-1}}|\mu_{1}(t)+\hat{\delta}_{1}(t)|}\right\} 
\end{array}\right\} \\
+b_{n}^{-2}\left\{ \frac{\sup_{|t|\le b_{n}^{-1}}|\mu_{3}(t)|\sup_{|t|\le b_{n}^{-1}}|\hat{\delta}_{2}(t)|}{\{\inf_{|t|\le b_{n}^{-1}}|\mu_{2}(t)|\}^{2}}+\frac{\sup_{|t|\le b_{n}^{-1}}|\hat{\delta}_{3}(t)|}{\inf_{|t|\le b_{n}^{-1}}|\mu_{2}(t)|}\right\} ^{2}\\
\times\left\{ 1+\frac{\sup_{|t|\le b_{n}^{-1}}|\hat{\delta}_{2}(t)|}{\inf_{|t|\le b_{n}^{-1}}|\mu_{2}(t)+\hat{\delta}_{2}(t)|}\right\} ^{2}\left\{ 1+\frac{\sup_{|t|\le b_{n}^{-1}}|\hat{\delta}_{1}(t)|}{\inf_{|t|\le b_{n}^{-1}}|\mu_{1}(t)|}+\frac{\{\sup_{|t|\le b_{n}^{-1}}|\hat{\delta}_{1}(t)|\}^{2}}{\inf_{|t|\le b_{n}^{-1}}|\mu_{1}(t)+\hat{\delta}_{1}(t)|}\right\} 
\end{array}\right)\\
 & =O_{p}\left(\frac{n^{-1}\log(1/b_{n})^{2}\{\inf_{|t|\le b_{n}^{-1}}|f^{{\rm ft}}(t)|\}^{-1}}{\min\left\{ \inf_{|t|\le b_{n}^{-1}}|f_{\epsilon}^{{\rm ft}}(t)|,\{\inf_{|t|\le b_{n}^{-1}}|f_{\nu}^{{\rm ft}}(t)|\}^{4}\{\inf_{|t|\le b_{n}^{-1}}|f^{{\rm ft}}(t)|\}^{3}b_{n}^{2}\right\} }\right),
\end{align*}
where the last step uses Lemma \ref{lem:delta}, $\frac{\inf_{|t|\le b_{n}^{-1}}|\mu_{1}(t)|}{\inf_{|t|\le b_{n}^{-1}}|f_{\epsilon}^{{\rm ft}}(t)|\inf_{|t|\le b_{n}^{-1}}|f^{{\rm ft}}(t)|}\ge1$,
$\frac{\inf_{|t|\le b_{n}^{-1}}|\mu_{2}(t)|}{\inf_{|t|\le b_{n}^{-1}}|f_{\nu}^{{\rm ft}}(t)|\inf_{|t|\le b_{n}^{-1}}|f^{{\rm ft}}(t)|}\ge1$,
and $\sup_{|t|\le b_{n}^{-1}}|\mu_{3}(t)|=O(1)$ under Assumption
(1).
\end{proof}
\begin{lem}
\label{lem:pi2} Under Assumption (1),
\begin{align*}
\sup_{|t|\le b_{n}^{-1}}E[|\Pi_{1}(t)|^{2}] & =O\left(\frac{\{\inf_{|t|\le b_{n}^{-1}}|f^{{\rm ft}}(t)|\}^{-2}}{\min\left\{ \{\inf_{|t|\le b_{n}^{-1}}|f_{\epsilon}^{{\rm ft}}(t)|\}^{2},\{\inf_{|t|\le b_{n}^{-1}}|f_{\nu}^{{\rm ft}}(t)|\}^{2}b_{n}^{2}\right\} }\right),
\end{align*}
which implies that for $s=0,1$, 
\[
\sup_{|t|\le b_{n}^{-1}}E[|(Y_{1}-cW_{1})^{s}\Pi_{1}(t)|]=O\left(\frac{\{\inf_{|t|\le b_{n}^{-1}}|f^{{\rm ft}}(t)|\}^{-1}}{\min\left\{ \inf_{|t|\le b_{n}^{-1}}|f_{\epsilon}^{{\rm ft}}(t)|,\inf_{|t|\le b_{n}^{-1}}|f_{\nu}^{{\rm ft}}(t)|b_{n}\right\} }\right).
\]
\end{lem}

\begin{proof}
The second statement follows by the first statement, the Cauchy-Schwartz
inequality and Assumption (1). The conclusion then follows by
\begin{align*}
 & \sup_{|t|\le b_{n}^{-1}}E[|\Pi_{1}(t)|^{2}]\le E\left[\sup_{|t|\le b_{n}^{-1}}\left|-\frac{\delta_{1,1}(t)}{\mu_{1}(t)}+{\rm i}\int_{0}^{t}\left\{ -\frac{\mu_{3}(s)\delta_{2,1}(s)}{\mu_{2}^{2}(s)}+\frac{\delta_{3,1}(s)}{\mu_{2}(s)}\right\} ds\right|^{2}\right]\\
\le & E\left[\left(\frac{\sup_{|t|\le b_{n}^{-1}}|\delta_{1,1}(t)|}{\inf_{|t|\le b_{n}^{-1}}|\mu_{1}(t)|}+b_{n}^{-1}\left\{ \frac{\sup_{|t|\le b_{n}^{-1}}|\mu_{3}(t)|\sup_{|t|\le b_{n}^{-1}}|\delta_{2,1}(t)|}{\{\inf_{|t|\le b_{n}^{-1}}|\mu_{2}(t)|\}^{2}}+\frac{\sup_{|t|\le b_{n}^{-1}}|\delta_{3,1}(t)|}{\inf_{|t|\le b_{n}^{-1}}|\mu_{2}(t)|}\right\} \right)^{2}\right]\\
= & O\left(\frac{E[\sup_{|t|\le b_{n}^{-1}}|\delta_{1,1}(t)|^{2}]}{\{\inf_{|t|\le b_{n}^{-1}}|\mu_{1}(t)|\}^{2}}+\frac{\{\sup_{|t|\le b_{n}^{-1}}|\mu_{3}(t)|\}^{2}E[\sup_{|t|\le b_{n}^{-1}}|\delta_{2,1}(t)|^{2}]}{b_{n}^{2}\{\inf_{|t|\le b_{n}^{-1}}|\mu_{2}(t)|\}^{2}}+\frac{E[\sup_{|t|\le b_{n}^{-1}}|\delta_{3,1}(t)|^{2}]}{b_{n}^{2}\{\inf_{|t|\le b_{n}^{-1}}|\mu_{2}(t)|\}^{2}}\right)\\
= & O\left(\frac{\{\inf_{|t|\le b_{n}^{-1}}|f^{{\rm ft}}(t)|\}^{-2}}{\min\left\{ \{\inf_{|t|\le b_{n}^{-1}}|f_{\epsilon}^{{\rm ft}}(t)|\}^{2},\{\inf_{|t|\le b_{n}^{-1}}|f_{\nu}^{{\rm ft}}(t)|\}^{2}b_{n}^{2}\right\} }\right),
\end{align*}
where the last step uses $\frac{\inf_{|t|\le b_{n}^{-1}}|\mu_{1}(t)|}{\inf_{|t|\le b_{n}^{-1}}|f_{\epsilon}^{{\rm ft}}(t)|\inf_{|t|\le b_{n}^{-1}}|f^{{\rm ft}}(t)|}\ge1$,
$\frac{\inf_{|t|\le b_{n}^{-1}}|\mu_{2}(t)|}{\inf_{|t|\le b_{n}^{-1}}|f_{\nu}^{{\rm ft}}(t)|\inf_{|t|\le b_{n}^{-1}}|f^{{\rm ft}}(t)|}\ge1$,
$\sup_{|t|\le b_{n}^{-1}}|\mu_{3}(t)|=O(1)$ and $E[\sup_{|t|\le b_{n}^{-1}}|\delta_{\iota,1}(t)|^{2}]<\infty$
for $\iota=1,2,3$ under Assumption (1).
\end{proof}
\begin{lem}
\label{lem:bias1} Under Assumptions (1) and (3), for $s,k=0,1$,
\[
E\left[(Y-W)^{s}\mathbb{K}^{(k)}\left(\frac{x-X}{b_{n}}\right)\right]=E\left[(Y-W)^{s}K^{(k)}\left(\frac{x-X^{*}}{b_{n}}\right)\right].
\]
\end{lem}

\begin{proof}
For $s,k=0,1$,
\begin{align*}
E\left[(Y-W)^{s}\mathbb{K}^{(k)}\left(\frac{x-X}{b_{n}}\right)\right] & =E\left[(Y-W)^{s}\left\{ \frac{1}{2\pi}\int e^{-\mathrm{i}t\left(\frac{x-X}{b_{n}}\right)}\frac{K^{\mathrm{ft}}(t)(-{\rm i}t)^{k}}{f_{\epsilon}^{\mathrm{ft}}(t/b_{n})}dt\right\} \right]\\
 & =E\left[(Y-W)^{s}\left\{ \frac{1}{2\pi}\int e^{-\mathrm{i}t\left(\frac{x-X^{*}}{b_{n}}\right)}K^{\mathrm{ft}}(t)(-{\rm i}t)^{k}dt\right\} \right]\\
 & =E\left[(Y-W)^{s}K^{(k)}\left(\frac{x-X^{*}}{b_{n}}\right)\right],
\end{align*}
where the second step follows from the independence between $\epsilon$
and $Y$, and the last step follows from the fact $\{K^{(k)}\}^{{\rm ft}}(t)=K^{\mathrm{ft}}(t)(-{\rm i}t)^{k}$
for $k=0,1$.
\end{proof}
\begin{lem}
\label{lem:bias2} Under Assumptions (2) and (3), for $k=0,1$,
\begin{align*}
b_{n}^{-(k+1)}E\left[K^{(k)}\left(\frac{x-X^{*}}{b_{n}}\right)\right] & =f^{(k)}(x)+\frac{b_{n}^{\alpha-k}}{(\alpha-k)!}\int K(u)\{f^{(\alpha)}(x+b_{n}\bar{u}_{f})-f^{(\alpha)}(x)\}u^{\alpha-k}du,\\
b_{n}^{-(k+1)}E\left[(Y-cW)K^{(k)}\left(\frac{x-X^{*}}{b_{n}}\right)\right] & =h_{c}^{(k)}(x)+\frac{b_{n}^{\alpha-k}}{(\alpha-k)!}\int K(u)\{h_{c}^{(\alpha)}(x+b_{n}\bar{u}_{h})-h_{c}^{(\alpha)}(x)\}u^{\alpha-k}du,
\end{align*}
for some $\bar{u}_{f}$ and $\bar{u}_{h}$ such that $\max\{|\bar{u}_{f}|,|\bar{u}_{h}|\}\le|u|$. 
\end{lem}

\begin{proof}
Since the arguments are similar, we focus on the second statement
when $k=1$, which follows from
\begin{align*}
 & b_{n}^{-2}E\left[(Y-cW)K^{\prime}\left(\frac{x-X^{*}}{b_{n}}\right)\right]=b_{n}^{-2}\int h_{c}(x^{*})K^{\prime}\left(\frac{x-x^{*}}{b_{n}}\right)dx^{*}\\
= & -b_{n}^{-1}\int h_{c}(x^{*})dK\left(\frac{x-x^{*}}{b_{n}}\right)=b_{n}^{-1}\int K\left(\frac{x^{*}-x}{b_{n}}\right)h_{c}^{\prime}(x^{*})dx^{*}\\
= & \int K(u)\underbrace{h_{c}^{\prime}(x+b_{n}u)}_{\sum_{l=0}^{\alpha-1}\frac{h_{c}^{(l+1)}(x)}{l!}b_{n}^{l}u^{l}+\frac{b_{n}^{\alpha-1}}{(\alpha-1)!}\{h_{c}^{(\alpha)}(x+b_{n}\bar{u}_{f})-h_{c}^{(\alpha)}(x)\}u^{\alpha-1}}du\\
= & h_{c}^{\prime}(x)+\frac{b_{n}^{\alpha-1}}{(\alpha-1)!}\int K(u)\{h_{c}^{(\alpha)}(x+b_{n}\bar{u}_{f})-h_{c}^{(\alpha)}(x)\}u^{\alpha-1}du,
\end{align*}
where the third step follows from the integration by parts and the
symmetry of the kernel function $K$, the fourth step follows from
the change of variables $u=(x^{*}-x)/b_{n}$, and the last step follows
from the property of the kernel function $K$ as in Assumption (3).
\end{proof}
\newpage{}

\newpage{}

\begin{figure}[h]
	\caption{\label{fig:simulation}Monte Carlo Simulation Results}
	\begin{centering}
		\includegraphics[scale=0.8]{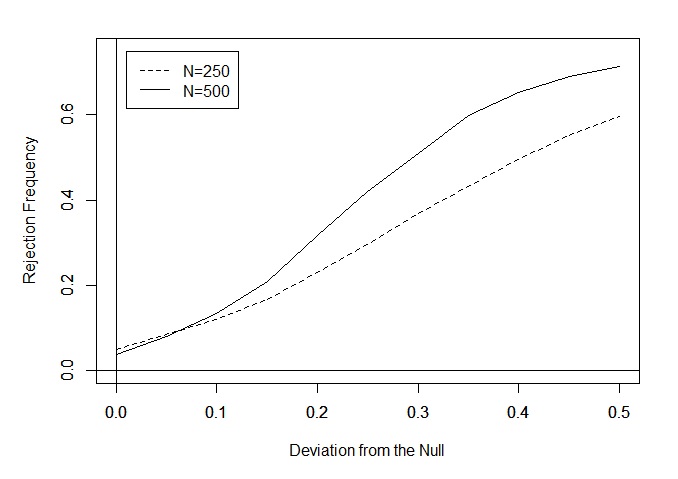}
		\par\end{centering}
	\centering{}%
	\begin{minipage}[t]{0.83\columnwidth}%
		Notes: The vertical axis measures the rejection frequency for $H_{0}:\theta_{1}\ge0$
		with the nominal size of 0.05. The horizontal axis measures the deviation
		$\delta\in[0.0,0.5]$ from the null hypothesis $H_{0}:\theta_{1}\ge0$.
		The dashed (respectively, solid) line indicates the results with $N=250$
		(respectively, $500$).%
	\end{minipage}
\end{figure}

\newpage{}

\begin{table}[h]
	\caption{\label{tab:summary_statistics}Summary Statistics of U.S. PSID for
		2013, 2015, 2017, and 2019. }
	\vspace{0.35cm}
	\begin{tabular}{ccccccc}
		\hline\hline
		Year & Log Income & Log Consumption && $X$ & $W$ & $Y$ \\
		\cline{1-3}\cline{5-7}
		2013 &10.682 &10.501 && 0.167 & 0.145 & 0.033\\
		&(1.086)&(0.762)&&(0.730)&(0.772)&(0.429)\\
		2015 &10.766 &10.537\\
		&(1.061)&(0.739)\\
		2017 &10.857 &10.578\\
		&(1.021)&(0.696)\\
		2019 &10.920 &10.640\\
		&(1.042)&(0.708)& \multicolumn{4}{r}{Observations = 5976}\\
		\hline\hline
	\end{tabular}
	\begin{centering}
		\bigskip{}
		\par\end{centering}
	\centering{}%
	\begin{minipage}[t]{0.83\columnwidth}%
		Notes: Displayed values are the sample means. Parentheses enclose
		sample standard deviations.%
	\end{minipage}
\end{table}

\end{document}